\documentclass{eptcs}

\usepackage{breakurl}
\usepackage{amssymb}
\usepackage{amsmath}
\usepackage{amsthm}
\usepackage{pepa}
\usepackage{graphicx}
\usepackage{rotating}
\usepackage{stmaryrd}
\usepackage{url}
\usepackage[mathscr]{eucal}
\usepackage{setspace}

\usepackage{tikz}
\usetikzlibrary{shapes.arrows}
\usetikzlibrary{patterns}

\newcommand{\publicationstatus}{}

\newcommand{\rl}[2]{\rule{#1}{0pt}\rule{0pt}{#2}}
\newcommand{\rsdef}[3]{#1' #3 #2}
\newcommand{\rsdefa}[2]{\rsdef{#1}{#2}{\!\sim\!}}

\newcommand{\calF}{\mathcal{F}}

\newcommand{\calX}{\mathcal{X}}
\newcommand{\calU}{\mathcal{U}}
\newcommand{\calY}{\mathcal{Y}}
\newcommand{\calN}{\mathcal{N}}

\newcommand{\smsync}[1]{\smash{\:\sync{#1}\:}}

\newcommand{\Assem}{\mathit{Assembler}}
\newcommand{\assemtime}{\mathit{atime}}
\newcommand{\taketime}{\mathit{ttime}}
\newcommand{\examtime}{\mathit{etime}}
\newcommand{\resttime}{\mathit{rtime}}

\newcommand{\Feed}{\mathit{Feed}}
\newcommand{\arr}{\mathit{arrivals}}
\newcommand{\dep}{\mathit{departures}}

\newcommand{\wt}{\mathit{wt}}
\newcommand{\wa}{\mathit{wa}}

\newcommand{\AProc}{\mathit{AProc}}
\newcommand{\Timer}{\mathit{Timer}}
\newcommand{\Machine}{\mathit{Machine}}
\newcommand{\Inspect}{\mathit{Inspect}}

\newtheorem{definition}{Definition}
\newtheorem{proposition}{Proposition}
\newtheorem{theorem}{Theorem}

\usepackage{tikz}
\usetikzlibrary{automata}

\newcommand{\ev}[1]{\underline{\mathrm{#1}}}
\newcommand{\sev}[1]{\overline{\mathrm{#1}}}
\newcommand{\ssev}[1]{\mathrm{#1}}

\newcommand{\Ev}{\mathcal{E}}

\newcommand{\ec}{\mathit{ec}}

\newcommand{\EC}{\mathit{EC}}

\newcommand{\iname}[1]{\mathit{#1}}

\newcommand{\const}{\mathit{const}}
\newcommand{\linear}{\mathop{\mathit{linear}}}

\renewcommand{\Ac}{\mathcal{A}}

\newcommand{\iv}{\mathop{\mathit{iv}}}

\newcommand{\ID}{\mathit{ID}}
\newcommand{\IN}{\mathit{IN}}
\newcommand{\IT}{\mathit{IT}}

\newcommand{\D}{\displaystyle}
\newcommand{\lts}[1]{\stackrel{#1}{\longrightarrow}}

\newcommand{\multi}{\mathop{\mathit{mult}}}

\newcommand{\ds}{\mathop{\mathrm{ds}}}
\newcommand{\st}{\mathop{\mathrm{st}}}

\newcommand{\assign}[1]{\llbracket #1 \rrbracket}
\newcommand{\eval}[1]{\llbracket #1 \rrbracket}
\newcommand{\V}{\mathcal{V}}

\newcommand{\ConSys}{\mathit{ConSys}}
\newcommand{\Con}{\mathit{Con}}

\newcommand{\CF}{\mathcal{F}}

\newcommand{\bbR}{\mathbb{R}}
\newcommand{\calH}{\mathcal{H}}

\newcommand{\blang}{\bigl\langle}
\newcommand{\brang}{\bigr\rangle}
\newcommand{\cf}[1]{\langle #1 \rangle}

\newcommand{\St}{\mathcal{S}}

\newcommand{\is}{\mathsf{is}}
\newcommand{\es}{\mathsf{ev}}

\newcommand{\vr}[1]{\mathbf{#1}}
\newcommand{\TC}[1]{\mathcal{TC}_{#1}}
\newcommand{\TD}[1]{\mathcal{TD}_{#1}}
\newcommand{\TS}[1]{\mathcal{TS}_{#1}}
\newcommand{\TDSHA}[1]{\mathcal{T}_{#1}}
\newcommand{\TDSHAtuple}[1]{(Q_{#1},\vr{X_{#1}},\TC{#1},\TD{#1},\TS{#1},\mathrm{init}_{#1},\mathcal{E}_{#1})}
\newcommand{\initial}[1]{\mathrm{init}_{#1}}
\newcommand{\modes}[1]{Q_{#1}}
\newcommand{\vrX}[1]{\mathbf{X_{#1}}}
\newcommand{\events}[1]{\mathcal{E}_{#1}}
\newcommand{\tprod}[1]{\otimes_{#1}}

\newcommand{\mode}[1]{q_{#1}}
\newcommand{\initmode}[1]{q_{#1}^{init}}
\newcommand{\stoich}[1]{\mathbf{s}_{#1}}
\newcommand{\rf}[1]{f_{#1}}
\newcommand{\exit}[1]{q_{1}^{#1}}
\newcommand{\enter}[1]{q_{2}^{#1}}
\newcommand{\priority}[1]{w_{#1}}
\newcommand{\guard}[1]{g_{#1}}
\newcommand{\res}[1]{r_{#1}}
\newcommand{\e}[1]{e_{#1}}
\newcommand{\inp}[1]{\mathrm{inp}_{#1}}

\newcommand{\act}{\mathit{act}}
\newcommand{\rs}{\mathit{res}}

\renewcommand{\TS}[1]{\mathcal{T\hspace*{-2.1pt}S}_{#1}}
\renewcommand{\TC}[1]{\mathcal{T\hspace*{-1.5pt}C}_{#1}}
\renewcommand{\TD}[1]{\mathcal{T\hspace*{-1.1pt}D}_{#1}}

\newcommand{\pc}{{\,:\,}}
\newcommand{\Sys}{\mathit{Sys}}

\newcommand{\W}{\mathcal{W}}
\newcommand{\calC}{\mathcal{C}}
\newcommand{\calS}{\mathcal{S}}
\newcommand{\prs}{\mathsf{pr}}
\newcommand{\ad}{\mathsf{sum}}
\newcommand{\match}{\mathsf{match}}
\newcommand{\out}{\mathsf{out}}

\newcommand{\lms}{\{\hspace*{-2.8pt}|\,}
\newcommand{\rms}{\,|\hspace*{-2.7pt}\}}

\newcommand{\true}{\mathit{true}}

\newcommand{\syncstar}{\smash{\sync{\rule{0pt}{2.8pt}}}\hspace*{-9.9pt}\raisebox{-2.8pt}[0pt][0pt]{$\scriptscriptstyle
*$}\hspace*{8pt}}

\newcommand{\Input}{\mathit{Input}}
\newcommand{\Output}{\mathit{Output}}
\newcommand{\Buffer}{\mathit{Buffer}}

\newcommand{\off}{\mathit{off}}

\newcommand{\ins}{\mathsf{in}}

\newcommand{\evs}{\mathsf{ev}}

\newcommand{\lsxrightarrow}[1]{\raisebox{-2.0pt}[0pt][0pt]{$\xrightarrow{\raisebox{0.2pt}[0pt][0pt]{$\scriptstyle{#1}$}}$}}

\newcommand{\mlsxrightarrow}[1]{\raisebox{-2.0pt}[0pt][0pt]{$\xrightarrow{\raisebox{1.6pt}[0pt][0pt]{$\scriptstyle{#1}$}}$}}

\newcommand{\lsexrightarrow}[1]{\raisebox{-2.0pt}[0pt][0pt]{$\xrightarrow{\raisebox{0.2pt}[0pt][0pt]{$\,\scriptstyle\ssev{#1}\,$}}$}}

\newcommand{\lxrightarrow}[1]{\raisebox{-1.0pt}[0pt][0pt]{$\xrightarrow{\raisebox{-1.5pt}[0pt][0pt]{$#1$}}$}}

\newcommand{\rmdefl}{\smash{\raisebox{-2.0pt}[0pt][0pt]{$\,\rmdef\,$}}}

\title{Stochastic HYPE: Flow-based modelling of stochastic hybrid
systems}
\author{Luca Bortolussi 
\institute{Department of Mathematics and Geosciences\\
University of Trieste}
\email{luca@dmi.units.it}
\and
Vashti Galpin \qquad\qquad Jane Hillston
\institute{Laboratory for Foundations of Computer Science\\
School of Informatics, University of Edinburgh}
\email{Vashti.Galpin@ed.ac.uk \quad Jane.Hillston@ed.ac.uk}
}
\def\titlerunning{Stochastic HYPE}
\def\authorrunning{L. Bortolussi, V. Galpin \& J. Hillston}

\begin{document}

\maketitle

\begin{abstract}
\noindent
Stochastic HYPE is a novel process algebra that models stochastic,
instantaneous and continuous behaviour. It develops the flow-based
approach of the hybrid process algebra HYPE by replacing non-urgent
events with events with exponentially-distributed durations and
also introduces random resets. The random resets allow for general
stochasticity, and in particular allow for the use of event durations
drawn from distributions other than the exponential distribution.
To account for stochasticity, the semantics of stochastic HYPE
target piecewise deterministic Markov processes (PDMPs), via
intermediate transition-driven stochastic hybrid automata (TDSHA)
in contrast to the hybrid automata used as semantic target for HYPE.
Stochastic HYPE models have a specific structure where the controller
of a system is separate from the continuous aspect of this system
providing separation of concerns and supporting reasoning.  A novel
equivalence is defined which captures when two models have the same
stochastic behaviour (as in stochastic bisimulation), instantaneous
behaviour (as in classical bisimulation) and continuous behaviour.
These techniques are illustrated via an assembly line example.

%\commentv{Needs to be reduced to 100 words}
%Stochastic HYPE is a novel process algebra that models stochastic,
%instantaneous and continuous behaviour. It develops the flow-based
%approach of the hybrid process algebra HYPE by replacing
%non-urgent actions with actions with stochastic duration and introducing
%random resets. In this
%paper, stochastic HYPE is defined formally, including both notions
%of \emph{well-defined} models which have specific properties and
%\emph{well-behaved} models that do not contain infinite sequences
%of instantaneous actions.
%
%The semantics of stochastic HYPE are given by structured operational
%semantics which define labelled transition systems that are then
%mapped to transition-driven stochastic hybrid automata, a subset
%of piecewise deterministic Markov processes in contrast to a previous
%definition that mapped compositionally to transition-driven stochastic
%hybrid transition automata. A behavioural equivalence based on
%bisimulation notions for both instantaneous transitions and stochastic
%transitions is defined, and its use is demonstrated by an assembly
%line example.  
%\commentv{Needs to be reduced to 100 words}
\end{abstract}

\section{Introduction}
\label{sec:intro}

In the last decade there has been increasing interest in capturing
stochastic behaviour in models of computational systems.  The
motivations for this work range from being able to capture probabilistic
algorithms, which gain efficiency through non-determinism, to
abstraction over an uncertain environment, where a complex variety of
possible responses may be represented as a probability distribution
rather than a single response.  In this paper we introduce
stochastic behaviour to HYPE, a formalism previously proposed to model
hybrid systems \cite{HYPE-journal}.  Our motivation is to capture uncertain
response events in a quantified manner.

Hybrid systems mix discrete control and continuous evolution, and can
arise in a number of natural and engineering contexts.  A number of
process algebras for modelling hybrid systems have emerged in recent
years including $\smash{\mathit{ACP}_{hs}^{srt}}$ \cite{BergBM:05a},
hybrid $\smash{\Chi}$ \cite{BeekBMRRS:06a}, $\phi$-calculus
\cite{RounRS:03a}
and HyPA \cite{CuijCR:05a}. 
\cite{KhadK:08a} shows
in her thesis that there are
substantial differences in the approaches taken by these process
algebras relating to syntax, semantics, discontinuous behaviour,
flow-determinism, theoretical results and availability of tools.
Nevertheless there are also strong similarities in their approaches to
capturing the continuous evolution of the system: the dynamic behaviour
of each continuous variable must be fully described by ordinary
differential equations (ODEs) given explicitly in the syntax of the
process algebra.  Furthermore, they all adopt a notion of state which
incorporates an evaluation of continuous variables.  

HYPE adopts a different approach which supports a finer-grained view on
the dynamics of the system.  As it is a process algebra, models are
constructed compositionally, but here each subcomponent is itself
constructed from a number of \emph{flows} or influences, represented
abstractly.  Ultimately each flow will correspond to a term in the ODEs
governing the evolution of a continuous variable, but at the process
algebra level the details of the mathematical form of this influence are
left unspecified.  This separation of concerns allows us to consider the
logical structure of the model (the process algebra description) as
distinct from any particular dynamic realisation.    We believe that the
use of flows as the basic elements of model construction has advantages
such as ease and simplification of modelling. This approach assists the
modeller, allowing them to identify smaller or local descriptions of the
model and then to combine these descriptions to obtain the larger
system.

Discrete changes within a system are captured as \emph{events}.
Each event may change the dynamic nature of a flow, possibly also
resetting the current value of continuous variables.  Each event
has a triggering condition which may be dependent on the current
value of the continuous variables.  In HYPE as presented by
\cite{HYPE-journal}, events are either \emph{urgent} or \emph{non-urgent}.
Urgent events are forced to occur as soon as their trigger becomes
true, and if multiple urgent events become true simultaneously, one
of them occurs, and the remaining events may or may not occur after
this event depending on whether their trigger condition remains
true.  In contrast, non-urgent actions may delay indefinitely.

HYPE can be viewed as belonging to the class of quantified process
algebras, including probabilistic, stochastic and timed process algebras
which support reasoning over models to quantitatively analyse the
behaviour of a system, for example considering the relative likelihood
of particular sub-behaviours, or the expected time until a condition is
met.  However, the inclusion of non-urgent events within the model left
some unspecified behaviour that could not be quantified and therefore
limited the extent to which quantified analysis could be applied.  
Thus in this paper we define an extension of HYPE, stochastic HYPE, in
which the non-urgent events of HYPE are replaced by stochastic events.
Such events remain non-urgent, as in not occurring immediately, but now
their delay is governed by an exponential distribution. 
Additionally through the use of random variables in resets, other
forms of stochasticity can be included in models, and timers can
be used to achieve random delays from an arbitrary distribution.  So the
non-urgent events of HYPE are effectively replaced in stochastic
HYPE by random events based on any distribution. The underlying
mathematical model used by stochastic HYPE is Transition Driven
Stochastic Hybrid Automata (TDSHA)
\cite{PA:Bortolussi:2010:HybridDynamicsStochProg:TCS}, based on
piecewise deterministic Markov process (PDMP) \cite{Davis93}.
We work with TDSHA because they are a better than operational match with
stochastic HYPE than PDMPs.

The contributions of this paper include a compositional language
to describe the dynamics of systems with stochastic, instantaneous
and continuous behaviour, in a modular fashion to ensure models are
straightforward to modify. This requires mapping to a mathematical
model of these three behaviour types to TDSHA, related to PDMP.
Furthermore, a new bisimulation is defined to enable formal reasoning
about differences and similarities in model behaviour. These
contributions are illustrated through a non-trivial example of an
assembly line.

The remainder of this paper is structured as follows.  In Section 2 we
illustrate the stochastic HYPE modelling language,
through the description of a simple network example, and the formal
syntax of stochastic HYPE.  The operational
semantics is defined in terms of \emph{configurations}
which capture the influences at play within the system rather than
explicit values of continuous variables, is presented in Section 3.
%Here we also consider the impact of adding stochastic events on the
%notion of well-definedness, which have been established to capture the
%HYPE models whose construction is ``reasonable'' in some sense.  
The
dynamic interpretation of stochastic HYPE models is given in terms of a
class of hybrid automata called Transition-Driven Stochastic Hybrid
Automata (TDSHA) \cite{PA:Bortolussi:2010:HybridDynamicsStochProg:TCS}
which can themselves be mapped to PDMPs
\cite{Davis93}.  The TDSHA formalism is presented in Section 4 and the
mapping from stochastic HYPE to TDSHA is defined in Section 5.  This
section also considers the impact of stochastic events on the notion
of well-behavedness.  A HYPE model is well-behaved if it cannot
execute an infinite number of simultaneous instantaneous events.
Introducing stochastic events to a model
that was previously not well-behaved, can lead to well-behavedness. Section 6 introduces
the main example of the paper and illustrates the expressiveness of
stochastic HYPE through an optimisation problem.
Semantic equivalences play an important
role in reasoning about process algebra models, and a major contribution
of this paper is a suitable bisimulation for stochastic HYPE models.
In Section 7, we show why existing bisimulations are not sufficient and
present a bisimulation
for stochastic HYPE that captures the notions that two models share the
same behaviours for all three behaviour types.
In Section 8, the use 
of this bisimulation and associated results are demonstrated on the
example from Section 6.
Then Section 9 presents related work and we conclude in Section 10.
Online appendices contain proofs and supplementary material.
% on TDHSA bisimulation.

Preliminary work on stochastic HYPE has been published by
\cite{BortBGH:11a}; here we define a more elegant mapping to TDSHA, add
random resets
and present new results on bisimulation equivalences.

\section{Stochastic HYPE Definition}
\label{sec:HYPE}

\begin{figure}[!t]
{\renewcommand{\arraystretch}{1.00}
$$\begin{array}{rcl}
\Input & \rmdef & 
 \sev{on}_{in}\pc(in,r_{in},\const).\Input +
 \sev{off}_{in}\pc(in,0,\const).\Input +\\
 & & \ev{full}\pc(in,0,\const).\Input +
 \ev{init}\pc(in,0,\const).\Input               \\
\Output & \rmdef &   
 \sev{on}_{out}\pc(out,-r_{out},\const).\Output +
 \sev{off}_{out}\pc(out,0,\const).\Output +\\
 & &  \ev{empty}\pc(out,0,\const).\Output +
 \ev{init}\pc(out,0,\const).\Output \\
\mathit{Drop} & \rmdef & \ev{init}\pc(f,0,\const).\mathit{Drop} +
                    \ev{fail}\pc(f,0,\const).\mathit{Drop} \\
\Timer & \rmdef & \ev{init}\pc(t,1,\const).\Timer \\
\\
\Con_{in} & \rmdef & \sev{on}_{in}.\Con_{in}' \qquad \Con_{in}' \:\rmdef\:
\sev{off}_{in}.\Con_{in} + \ev{full}.\Con_{in}\\
\Con_{out} & \rmdef & \sev{on}_{out}.\Con_{out}' \quad\, \Con_{out}' \:\rmdef\:
\sev{off}_{out}.\Con_{out} + \ev{empty}.\Con_{out}\\
\Con_\mathit{fail} & \rmdef & \ev{fail}.\Con_\mathit{fail} \\
& & \\
\Buffer & \rmdef & (\Input \syncstar \Output \syncstar \mathit{Drop}
\syncstar \Timer) \syncstar \ev{init}.(\Con_{in} \parallel
\Con_{out}
\parallel \Con_\mathit{fail}) \\
\end{array}$$

{\renewcommand{\arraystretch}{1.2}
$$\begin{array}{rclrcl}
\V & = & 
\multicolumn{4}{c}{\{B,T,C,D\} \quad
\iv(in) \: = \: B \quad \iv(out) \: = \: B \quad
\iv(f) \: = \: B \quad \iv(t) \: = \: T} \\
\\
\ec(\ev{init}) & = & \multicolumn{4}{l}{(true,\;\rsdefa{B}{b_0} \:\wedge\:
\rsdefa{T}{0} \:\wedge\: \rsdefa{C}{0}
\:\wedge\: \rsdefa{D}{\mathrm{ln\ }\calN(\Delta,\xi)})} \\
\ec(\ev{fail}) & = & \multicolumn{4}{l}{(T=C\!+\!D,
\;\rsdefa{C}{T} \:\wedge\: \rsdefa{D}{\mathrm{ln\ }\calN(\Delta,\xi)})
\:\wedge\: \rsdefa{B}{B\!-\!\calU(0,B)})}\\
\ec(\ev{full}) & = & (B = max_B,true) \qquad \qquad & 
\ec(\ev{empty}) & = & (B = 0,true) \\
\ec(\sev{on}_{in}) & = & (k_{in}^{on},\;true) \quad & \ec(\sev{off}_{in})
& = & (k_{in}^{\off},\;true) \\
\ec(\sev{on}_{out}) & = & (k_{out}^{on},\;true) \quad &
\ec(\sev{off}_{out}) & = & (k_{out}^{\off},\;true) \\
\end{array}$$}
}
\caption{Simple network node model in stochastic HYPE.}\label{simpleNetwork}
\end{figure}

In this section we present the definition  of stochastic HYPE by
means of a small example.  More details about the process algebra
HYPE can be found in an earlier paper \cite{HYPE-journal}.  We
consider a basic model of a network node with a single buffer, which
can either receive packets from an input channel or send packets
to an output channel.  We assume that the number of packets that
travel through the node and that are stored in the buffer is large,
hence we describe them as a fluid quantity.  Received packets are
stored in the buffer, waiting to be sent. We allow reception and
sending of packets to happen concurrently, but enforcing a mutually
exclusive send/receive policy can be easily done as well.  We also
assume that uplinks and downlinks are not always working, but they
are activated and deactivated depending on the availability of a
connection.  These events are described as stochastic, with activation
and deactivation times governed by exponential distributions.
Incoming traffic has to be stopped if the buffer becomes full and
outgoing traffic has to be stopped when the buffer is empty. The
buffer also shows intermittent error behaviour in that after the
passage of a random time period, (drawn from a log normal distribution
with fixed mean and variance), it drops some of its packets (this
quantity is determined uniformly over the number of packets).

HYPE modelling is centred around the notion of \emph{flow}, which
represents some sort of influence continuously modifying one quantity
of the system, described by a variable taking continuous values.
Both the strength and form of a flow can be changed by the happening
of \emph{events}.  In our example, there are two flows modifying
the buffer level, modelled by the continuous variable $B$, namely
reception and sending of packets.  The inflow of packets is modelled
by the $\Input$ \emph{subcomponent} shown in Figure~\ref{simpleNetwork}.
Each subcomponent is a summation of prefixes consisting of an
\emph{event} $\ssev{a}$ followed by an \emph{activity} $\alpha$.
\emph{Events} ($\ssev{a} \in \Ev_d\cup\Ev_s$) are actions which
trigger discrete changes.  They can be caused by a controller, which
triggers them depending on the global state of the system, specifically
on values of variables (instantaneous events $\ev{a}\in\Ev_d$), or
happen at exponentially distributed time instants (stochastic events
$\sev{a}\in\Ev_s$). In the description of the example, non-exponential
durations are mentioned but these are constructed explicitly with
timers. We show later in the paper how we use a syntactic shorthand
for events that have a duration from any distribution.

The $\Input$ subcomponent reacts to three
different events: two stochastic ones, activation $\sev{on}_{in}$
and deactivation $\sev{off}_{in}$ of the uplink, and one instantaneous
event, $\ev{full}$, triggered when the buffer level is at its maximum
capacity. The event $\ev{init}$ serves to initialize the system.

\emph{Activities} ($\alpha \in \Ac$) are \emph{influences} on the
evolution of the continuous part of the system and define flows.
An activity is defined as a triple and can be parameterised by a
set of variables, $\alpha(\W) = (\iota,r,I(\W))$.  This triple
consists of an \emph{influence name} $\iota$, a rate of change (or
\emph{influence strength}) $r$ and an \emph{influence type name}
$I(\W)$ which describes how that rate is to be applied to the
variables involved, or the actual form of the flow\footnote{For
convenience, we will use $I$ for $I(\W)$ when $\W$ can be inferred.}.
In $\Input$, there are two distinct activities: $(in,r_{in},\const)$
and $(in,0,\const)$.  The first one gives the effect of the incoming
link being active on the buffer level $B$: the influence name is
$in$, which uniquely identifies the effect of the input link on
$B$, $r_{in}$ is the strength of the inflow of packets (here it can
be seen as the amount of received data per time unit), which is
associated with the function $\const$.  The second activity captures
the effect of the link being down. It again affects influence $in$,
but has strength $0$ and the form it takes is again $\const$.

The interpretation of influence types will be specified separately,
in order to experiment with different functional forms of the packet
inflow without modifying the subcomponent. In this case, we will
obviously interpret $\const$ as the constant function 1.  Hence,
in HYPE we separate the description of the logical structure of
flows from their mathematical interpretation.

A second subcomponent affecting buffer level is the output component
$\Output$ in Figure~\ref{simpleNetwork}, modelling the sending of
packets, which is defined similarly to $\Input$.  The third
subcomponent $\mathit{Drop}$ has no effect on the contents of the
buffer, since its influence strength is always 0, but it does
introduce the event $\smash{\ev{fail}}$ that corresponds to the
dropping of packets. This event will be described below.  Finally,
since the dropping of packets occurs with a frequency, a time
subcomponent and a time variable are required to describe the passing
of time.  It is not affected by any event apart from $\ev{init}$.
These subcomponents can be combined to give the overall uncontrolled
system
\begin{eqnarray*}
\Input \syncstar \Output \syncstar \mathit{Drop}
\syncstar \Timer
\end{eqnarray*}
Here {$\syncstar$} represents parallel cooperation where all shared
events must be synchronised (we use $\!\smsync{L}\!$ to specify that the
events in $L$ are to be synchronised on and ${\parallel}$ when no events are to
be synchronised on).  This cooperation of subcomponents is
called the \emph{uncontrolled system} because it only specifies how
flows react to events, without imposing any causal or temporal
constraints on them.

Causality on events, reflecting natural constraints or design
choices, is specified separately in the controller $\Con$. For the
system at hand, the controller is defined as shown in
Figure~\ref{simpleNetwork}.
The subcontroller
$\Con_{in}$ models the fact that the reception of packets can be
turned off only after being turned on. Furthermore, it describes
termination of the input if the buffer becomes full, but only if
the uplink is active. $\Con_{out}$ is similar.
The $\smash{\ev{fail}}$ event happens without reference to other
events, and is determined by a timer as will be shown below.  However,
a simple controller is added for this event so that the same events
appear both in the controller and the uncontrolled system.

The uncontrolled system and the controller are then combined together
to define the \emph{controlled system} $\Buffer$.

\emph{Controllers} have only event prefixes, and we need to specify
when events are activated. This is achieved by assigning to each
event a set of event conditions, which differ between
\emph{stochastic} and \emph{deterministic} events.

Deterministic events $\ev{a}\in \Ev_d$ happen when certain conditions
are met by the system.  These \emph{event conditions} are specified
by a function $ec$, assigning to each event a \emph{guard} or
\emph{activation condition} (a boolean combination of predicates
depending on system variables, stating when a transition can fire)
and a \emph{reset} (specifying how variables are modified by the
event).  For example, $\ec(\ev{full}) = (B = max_B,true)$ states
that the uplink is shut down when the buffer reaches its maximum
capacity $max_B$, and no variable is modified. An event condition
that involves resets is that of the event $\ev{fail}$\footnote{In
the original definition of HYPE, resets were deterministic and were
defined using standard functional notation. This notation must be
modified to allow for random resets and is defined formally later
in this section.}.
\begin{eqnarray*}
\ec(\ev{fail}) & = & {(T=C\!+\!D,
\;\rsdefa{C}{T} \:\wedge\: \rsdefa{D}{\mathrm{ln\ }\calN(\Delta,\xi)})
\:\wedge\: \rsdefa{B}{B\!-\!\calU(0,B)})}
\end{eqnarray*}
The activation condition requires the time variable to be the sum
of $C$ and $D$, and when this is true, $C$ is assigned the value
of the current time, and a new random duration $D$ is drawn from a
log normal distribution with mean $\Delta$ time units and variance $\xi$
thus determining how long until the next failure.  The value of $B$
is also decreased by a random variable drawn from the uniform
distribution for all non-negative values less than or equal to $B$
and this modification of variable $B$ represents the dropping of
packets.

The symbol $\sim$ indicates that there may be random values in the
reset. When no distribution is indicated, such as $\rsdefa{C}{T}$,
this is equivalent to $C'=T$. Multiple resets are combined using
$\wedge$ indicating that all resets must occur.  The notation $V'$
is used to denote the value of $V$ after the reset occurred.
Deterministic events in HYPE are \emph{urgent}, meaning that they
fire as soon as their guard becomes true.

Stochastic events $\sev{a}\in \Ev_s$ have an event condition composed
of a stochastic rate (replacing the guard of deterministic events)
and a reset the same as above. For instance, $\ec(\sev{on}_{in}) =
(k_{in}^{on},true)$ states that the duration of packet reception
is a stochastic event happening at times exponentially distributed
with constant rate $k_{in}^{on}$.  Rates define exponential
distributions and can be functions of the variables of the system.
We choose to restrict delays for stochastic events to those that
are exponential distributed to avoid the issue of residual clocks
when there is no explicit timer associated with the delays. The
event condition for $\ev{fail}$ shows how delays with other distributions
can be modelled with the use of timers, and we generalise this
approach later.

We also need to link each influence with an actual variable. This
is done using the function $\iv$. For the example,
$\iv(in)\!=\!\iv(out)\!=\!\iv(f)\!=\!B$, where $B$ is the buffer
level.  The same variable can be modified by many influences. Note
also that in order to model situations in which an event modifies
the flow of more than one variable (in the example, the activation
of a link might modify, for instance, battery consumption), we can
simply define more subcomponents and combine them into structured
components, synchronizing them on shared events.

Finally, we need to interpret influence types, mapping them to
proper functions. In the example, we set $\assign{\const} = 1$, so
that $\const$ defines a constant flow. However, influence types can
be mapped to linear and non-linear functions as well, as shown in
other HYPE examples \cite{GHB08a,HYPE-journal}.

In the preceding informal discussion, we have introduced the main
constituents of a stochastic HYPE model including the combination
of flow components with a controller component, variables, association
between influences and variables, conditions that specify when
events occur, and definitions for the influence type functions.  To
understand the dynamics of this system, we need to derive ODEs to
describe how the variables change over time, and to further specify
how the discrete and the continuous dynamics interact.  We will do
this by defining an operational semantics, which will specify
qualitatively the behaviour of a controlled system, and then mapping
the so-obtained labelled transition system into a special class of
Stochastic Hybrid Automata.  Before that we give the formal definition
of stochastic HYPE.  In the rest of the paper, $\V$ is a set or
tuple of variables with $\W \subseteq \V$ denoting an arbitrary
subset of $\V$.

\begin{definition}
\label{def:sHYPE}
A \emph{stochastic HYPE model} is a tuple 
$(\ConSys,\V,\IN,\IT,\Ev_d,\Ev_s,\Ac,\ec,\iv,\EC,\ID)$ where
\begin{itemize}
\item $\ConSys$ is a well-defined controlled system as defined below.
\item $\V$ is a finite set of variables.
\item $\IN$ is a set of influence names and $\IT$ is a set of influence
type names.
\item $\Ev_d$ is the set of instantaneous events of the form $\ev{a}$ and $\ev{a}_i$.
\item $\Ev_s$ is the set of stochastic events of the form  $\sev{a}$ and
$\sev{a}_i$.
\item $\Ev = \Ev_d \cup \Ev_s$ is the set of all events for which the 
notation
$\ssev{a}$ and $\ssev{a}_i$ is used.
\item $\Ac$ is a set of activities of the form
$\alpha(\W) = (\iota,r,I(\W)) \in (\IN \times
\mathbb{R} \times \IT)$ where $\W \subseteq \V$.
\item $\ec:\Ev \rightarrow \EC$ maps events to event conditions. Event
conditions are pairs of activation conditions and resets. 
\begin{itemize}
\item Activation
conditions for the instantaneous events in $\Ev_d$ are formulae  with
free variables in $\V$ and for the stochastic
events in $\Ev_s$, they are functions
$f:\bbR^{|\V|}\rightarrow\bbR^+$.
\item Resets are conjunctions of formulae of the form 
$\rsdefa{V}{\theta(\V)}$ where
$\theta(\V)$ is an expression that may also include random
variables which are described by $\calX(p_1,\ldots,p_m)$ where $\calX$
is a distribution and the $p_i$ are the parameters for the distribution
which can themselves be expressions involving variables and random variables.
%The notation $\rsdef{V}{\Theta(\V)}{=}$ is used 
%once the random values have been determined. \newline
%\commentv{Check whether this notation is ever used.}
\end{itemize}
\item $\iv: \IN \rightarrow \V$ maps influence names to
variable names.
\item $\EC$ is a set of event conditions.
\item $\ID$ is a collection of definitions consisting
of a real-valued function for each influence type name
$\assign{I(\W)} = f(\W)$ where the variables in
$\W$ are from $\V$.
\item $\Ev$, $\Ac$, $\IN$ and $\IT$ are pairwise disjoint.
\end{itemize}
\end{definition}

In \cite{HYPE-journal}, \emph{well-defined HYPE models} are introduced
and the rationale behind the definition is to enforce a policy in
the way models are specified, forcing the modeller to separate the
description of the logical blocks constituting a HYPE model.
Subcomponents are ``flat'' (i.e. they always call themselves
recursively) and there is a one-to-one correspondence between
influences and subcomponents: each subcomponent describes how a
specific influence is modified by events. Furthermore, synchronization
must involve all shared events.  This notion extends straightforwardly
to stochastic HYPE models, and here we choose to define stochastic
HYPE systems in their well-defined form immediately rather than
having separate definitions.

\begin{definition}
A \emph{well-defined controlled system} has the form $\Sigma
\smash{\syncstar} \ev{init}.\Con$ where $\Sigma$ is a well-defined
uncontrolled system and $\Con$ is a well-defined controller, and
the synchronisation is on all shared events, in particular $\ev{init}$,
the initial event which must occur first (and has $\true$ as its
activation condition).  Furthermore, all events that appear in the
controller must appear in the uncontrolled system and \emph{vice versa}.
The set of well-defined controlled systems is denoted by $\calC$.

A \emph{well-defined controller} is defined by the two-level grammar
$M ::= \ssev{a}.M\ |\ 0 \ |\ M + M$ and $\Con ::= M\ |\ \Con
\smash{\sync{L}}\Con$ with $\ssev{a} \in \Ev$
and with $L \subseteq \Ev$.

The \emph{well-defined uncontrolled system} consists of well-defined
subcomponents in cooperation over shared events where each subcomponent
appears at most once and $\W_i \subseteq \V$.
\begin{eqnarray*}
\Sigma & \rmdef & S_1(\W_1) \syncstar \ldots \syncstar S_s(\W_s)
\end{eqnarray*}

\emph{Well-defined subcomponents} represent the uncontrolled
capabilities of the system and each is a choice over events
such that $\ssev{a}_j \neq \ssev{a}_k$ for $j \neq k$,
$\ssev{a}_j \neq \ev{init}$ for all $j$, and 
$\W \subseteq \V$.
\begin{eqnarray*} S(\W)
& \rmdef & {\sum}_{j=1}^n \ssev{a}_j{:}(\iota,r_j,I_j(\W)).S(\W) +
\ev{init}{:}(\iota,r,I(\W)).S(\W)
\end{eqnarray*}
\end{definition}

A subcomponent is ready to react whenever any of its events'
activation condition becomes true or completes, after which the
influence associated with that event comes into force, replacing
any previous influence. By considering all the influences mapped
to a particular variable for a particular configuration of the
system, an ODE can be constructed using the definitions in $\ID$
to describe the evolution of that variable whenever the system is
in that configuration. We assume well-defined stochastic HYPE models in
the rest of this paper. With the syntax of stochastic HYPE models
defined, the next three sections show how the dynamic behaviour of
stochastic HYPE models can be defined before illustrating this dynamic
behaviour with an substantial example of an assembly line.

\section{Operational Semantics}
\label{sec:SOS}

We now introduce the behaviour of stochastic HYPE models by defining
an appropriate semantics.  We will proceed in a different way to
that presented in the preliminary research \cite{BortBGH:11a}:
instead of mapping each subcomponent and each controller to a
stochastic hybrid automaton, and then composing them together with
a product construction, we will construct a labelled (multi)transition
system (LTS), similarly to \cite{Hill96}, and then map the LTS to
a stochastic hybrid automaton.

To construct the LTS, we need a notion of state. Proceeding in the
same manner as HYPE \cite{HYPE-journal}, states will record for
each influence its current strength and influence type. States
essentially capture the continuous behaviour, while the structure
of the LTS describes the discrete behaviour, both instantaneous and
stochastic.

\begin{definition}
An \emph{operational state} is a function $\sigma:\IN
\rightarrow (\mathbb{R} \times \IT)$. The set of all operational
states is $\St$.  A \emph{configuration} consists of a controlled
system together with an operational state $\blang \ConSys, \sigma
\brang$ and the set of configurations is $\CF$.
\end{definition}

In the following, we will usually refer to `operational states' as
`states'.  States can also be viewed as a set of triples
$(\iota,r,I(\W))$, where there is at most one triple containing
$\iota$. In this form, it is evident that states are simply collections
of enabled activities.  We stress that states describe flows, not
the values of continuous variables. As such, the LTS semantics of
stochastic HYPE is different to that of stochastic hybrid automata,
similarly to the way that the LTS semantics of HYPE differs from
those of hybrid automata \cite{HYPE-journal}.

The operational semantics give a labelled multitransition system
over configurations $(\CF,\Ev,{\rightarrow} \subseteq \CF \times
\Ev \times \CF)$. As customary, we  write $E\;
\lxrightarrow{\mathrm{\scriptstyle \:a\:}}\; F$, for $E, F \in \CF$.
The rules are given in Figure~\ref{opersem} and are essentially
those of non-stochastic HYPE\@. The only formal difference is that
we consider in the rules both instantaneous and stochastic events,
treating them in the same way. The fact that the operational semantic
rules are essentially the same is because the LTS describes the
syntactic structure of the target hybrid automaton.  The proper
dynamics will be attached to a (stochastic) HYPE model by converting
its LTS to a (stochastic) hybrid automaton. This second step of the
semantics construction differs significantly from that used for
HYPE \cite{HYPE-journal}.

\begin{figure}[t]
\begin{center}
{\renewcommand{\arraystretch}{0.5}
$\begin{array}{ll}
\textbf{Prefix with influence:} & \textbf{Prefix without influence:} \\ \\
\D \frac{}{\blang \ssev{a}:(\iota,r,I).P,
           \sigma \brang
           \lsexrightarrow{a} \blang P, \sigma[\iota \mapsto
           (r,I)] \brang}\;(\ssev{a} \in \Ev)  &
\D \frac{}{\blang \ssev{a}.P, \sigma \brang
   \lsexrightarrow{a} \blang P, \sigma \brang}\;(\ssev{a} \in \Ev) \\
\\
\\
\textbf{Choice:} & \textbf{Constant:} \\
\\
\D \frac{\blang P, \sigma \brang \lsexrightarrow{a} \blang P', \sigma' \brang}
  {\blang P + Q, \sigma \brang \lsexrightarrow{a} \blang P', \sigma' \brang}
\qquad
\D \frac{\blang Q, \sigma \brang \lsexrightarrow{a} \blang Q', \sigma' \brang}
  {\blang P + Q, \sigma \brang \lsexrightarrow{a} \blang Q', \sigma' \brang}
\qquad \qquad
&
\D \frac{\blang P, \sigma \brang \lsexrightarrow{a} \blang P',\sigma' \brang}
      {\blang A, \sigma \brang \lsexrightarrow{a} \blang P', \sigma' \brang}      (A \rmdefl P) \\
\\
\\
\multicolumn{2}{l}{
\hspace*{-0.16cm}
\begin{array}{l}
\textbf{Cooperation without synchronisation:} \\
\\
\D \frac{\blang P, \sigma \brang \lsexrightarrow{a} \blang P', \sigma'
\brang}
      {\blang P \sync{L} Q, \sigma \brang \lsexrightarrow{a}
       \blang P' \sync{L} Q, \sigma' \brang} (\ssev{a} \not\in L) \qquad
\D \frac{\blang Q, \sigma \brang \lsexrightarrow{a} \blang Q', \sigma' \brang}
      {\blang P \sync{L} Q, \sigma \brang \lsexrightarrow{a}
       \blang P \sync{L} Q', \sigma' \brang} (\ssev{a} \not\in L) \\
\\ \\
\textbf{Cooperation with synchronisation:} \\
\\
\D \frac{\blang P, \sigma \brang \lsexrightarrow{a} \blang P', \tau \brang
 \quad \blang Q, \sigma \brang \lsexrightarrow{a} \blang Q', \tau' \brang}
 {\blang P \sync{L} Q, {\sigma} \brang \lsexrightarrow{a}
 \blang P' \sync{L} Q', \Gamma(\sigma,\tau,\tau') \brang}
 (\ssev{a} \in L, \Gamma \text{ defined})
\end{array}}
\end{array}$}
\end{center}

\caption{Operational semantics for stochastic HYPE}
\label{opersem}
\end{figure}

Inspecting the rules of Figure \ref{opersem}, we can see that the
only rules updating the state are  Prefix with influence and
Cooperation with synchronisation.

In this latter case, in particular, we need to enforce consistency
in the way influences are updated by the cooperating components.
This is done by the function $\Gamma$.  The updating function
$\sigma[\iota\mapsto(r,I)]$ is defined as
\[ \sigma[\iota\mapsto(r,I)](x)= \begin{cases}
(r,I) & \text{if\ } x=\iota \\
\sigma(x) & \text{otherwise.}
\end{cases} \]
The notation $\sigma[u]$ will also be used for an update, with
$\sigma[u_1\ldots u_n]$ denoting $(\ldots((\sigma[u_1])[u_2])\ldots)[u_n]$.  The
partial function $\Gamma:\St \times \St \times \St \rightarrow \St$
is defined as follows.
\begin{eqnarray*}
{(\Gamma(\sigma,\tau,\tau'))}(\iota) & = &
\begin{cases}
\tau(\iota)  & \text{if $\sigma(\iota)=\tau'(\iota)$},\\
\tau'(\iota) & \text{if $\sigma(\iota)=\tau(\iota)$},\\
\text{undefined} & \text{otherwise}.
\end{cases}
\end{eqnarray*}
$\Gamma$ is undefined if both agents in a cooperation try to update
the same influence. See \cite{HYPE-journal} for further details.
In the following, we will also need some additional notions.

\begin{definition}
Given a controlled system $P$, define structurally the following sets:
\begin{itemize}
\item the \emph{set of events}, $\evs(P)$:
$\evs(\ssev{a}\pc\alpha.S) = \{ \ssev{a} \}$,  $\evs(\ssev{a}.S) = \{
\ssev{a} \}$, $\evs(S_1+S_2) = \evs(S_1) \cup \evs(S_2)$,\linebreak  $\evs(P_1
\smash{\sync{L}} P_2) = \evs(P_1) \cup \evs(P_2)$;
\item the \emph{set of influences}, $\ins(P)$:
$\ins(\ssev{a}\pc(\iota,r,I(\W)).S) = \{ \iota \}$ $\ins(\ssev{a}.S)
= \emptyset$, $\ins(S_1+S_2) = \ins(S_1) \cup \ins(S_2)$, $\ins(P_1
\smash{\sync{L}} P_2) = \ins(P_1) \cup \ins(P_2)$; \item the \emph{set
of prefixes}: $\prs(\ssev{a}\pc\alpha.S) = \{ \ssev{a}\pc\alpha
\}$, $\prs(\ssev{a}.S) = \emptyset$, $\prs(S_1+S_2) = \prs(S_1)
\cup \prs(S_2)$, $\prs(P_1 \smash{\sync{L}} P_2) = \prs(P_1) \cup
\prs(P_2)$.  \end{itemize} \end{definition}

\begin{definition}
The \emph{derivative
set} of a controlled system $P$, $\ds(P)$, is defined as the
smallest set satisfying
\begin{itemize}
\item if $\cf{P,\sigma} \lsexrightarrow{init} \cf{P',\sigma'}$ then
$\cf{P',\sigma'} \in \ds(P)$
\item if $\cf{P',\sigma'} \in \ds(P)$ and $\cf{P',\sigma'}
\lsexrightarrow{a}
\cf{P'',\sigma''}$ then $\cf{P'',\sigma''} \in \ds(P)$.
\end{itemize}
Furthermore, we indicate with $\st(P) = \{ \sigma \mid \cf{Q,\sigma}
\in \ds(P) \}$ the set of \emph{states} of the derivative set of a
controlled system $P$.
\end{definition}

In the $\Buffer$ model, there are four states in the LTS, corresponding
to the four possible combinations of activation of input and output
channels. The influences $f$ and $t$ are the same in all states.
For instance, the state in which the input link is active
and the output is not active is
\begin{eqnarray*}
\sigma & = & \{ in \mapsto (r_{in},\const), \:\: 
                     out \mapsto (0,\const), \:\:
                     f \mapsto (0,\const), \:\: 
                     t \mapsto (1,\const)\} 
\end{eqnarray*}

Propositions 5.1 to 5.5 proved for HYPE \cite{HYPE-journal}, hold
for well-defined stochastic HYPE models. In particular, the function
$\Gamma$ used in the Cooperation with synchronisation rule, is
always defined for well-defined stochastic HYPE systems. The $\Buffer$
system is well-defined, and hence $\Gamma$ is always defined.

\section{Transition-Driven Stochastic Hybrid Automata}
\label{sec:TDSHA}

Because stochastic HYPE considers three distinct types of behaviour,
whereas HYPE considers only two, we need a different target for our
semantics. The semantics for HYPE are provided by hybrid automata which
covers continuous and instantaneous behaviour but not stochastic
behaviour.
We  now present Transition-Driven Stochastic Hybrid Automata,
introduced by Bortolussi and Policriti
\cite{PA:Bortolussi:2010:HybridDynamicsStochProg:TCS,BorP11},
a formalization of stochastic hybrid automata putting emphasis on
transitions, which can be either discrete (corresponding to
instantaneous or stochastic jumps) or continuous (representing flows
acting on system variables).  This formalism can be seen as an
intermediate layer in defining the stochastic hybrid semantics of
stochastic HYPE\@. In fact, TDSHA are themselves mapped to Piecewise Deterministic
Markov Processes~\cite{Davis93}, so that their dynamics can be
formally specified in terms of the latter.  Due to space constraints,
we will not provide a formal treatment of this construction, and
refer the reader papers by Bortolussi and Policriti
\cite{PA:Bortolussi:2010:HybridDynamicsStochProg:TCS,BorP11}
and \cite{BortB:12} for further details. We choose to work with TDSHA
rather than PDMPs because their definition is more operational in
nature, and hence this leads to a more straightforward mapping.

%\newline \commentv{The issue of a finite number of instantaneous
%transitions which can be mapped to a single transition in a PDMP
%is totally elided here, and only briefly discussed at the start of
%the section on well-behavedness -- is this OK?}

\begin{definition}\label{def:TSHS}
A Transition-Driven Stochastic Hybrid Automaton (TDSHA) is a tuple
$\TDSHA{} =
(Q,\vr{X},\TC,\TD,$ $\TS,\mathrm{init},\mathcal{E})$, where
\begin{itemize}
\item $\modes{}$ is a finite set of \emph{control modes}.

\item $\vrX{} = \{X_1,\ldots,X_n\}$ is a set of real valued
\emph{system variables} with the time derivative of $X_j$ denoted
by $\dot{X_j}$ and the value of $X_j$ after a change of mode denoted
by $X_j'$.

\item $\TC{}$ is the \emph{multiset} of \emph{continuous transitions
or flows}, whose elements $\tau$ are triples
$(\mode{\tau},\stoich{\tau},\rf{\tau})$, where $\mode{\tau}\in
\modes{}$ is a mode, $\stoich{\tau}$ is a vector of size $|\vrX{}|$,
and $\rf{\tau}:\bbR^n\rightarrow \bbR$ is a \emph{Lipschitz continuous}
function.

\item $\TD{}$ is the \emph{set} of \emph{instantaneous transitions},
whose elements $\delta$ are tuples of the form
$(\exit{\delta},\enter{\delta},\guard{\delta},\res{\delta},$
$\priority{\delta},\e{\delta})$.  The transition goes from mode
$\exit{\delta}$ to mode $\enter{\delta}$.

\begin{itemize}
\item The \emph{guard} $\guard{\delta}$ is a first-order formula
with free variables from $\vr{X}$, representing the \emph{closed
set} $G_{\delta} = \{\vr{x}\in\bbR^n~|~\guard{}[\vr{x}]\}$.  \item
The \emph{reset} $\res{\delta}$ is a conjunction of atoms of the
form $X'=r(\mathbf{X},\mathbf{W})$ where
$r:\mathbb{R}^n\times\mathbb{R}^h\rightarrow\mathbb{R}$ is the reset
function of $X$ dependent on the variables $\mathbf{X}$ as well as
a random vector $\mathbf{W}$. Variables not appearing in the reset
remain unchanged and a reset with value $\true$ is the identity
function on each variable.

\item The \emph{weight (priority)} of the edge is $\priority{\delta}\in
\bbR^+$ and is used to solve non-determinism among two or more
active transitions.

\item The \emph{label} of the edge is $\e{\delta}\in\events{}$.
\end{itemize}

\item $\TS{}$ is the \emph{multiset} of \emph{stochastic transitions},
whose elements $\eta$ are tuples of the form $\eta =
(\exit{\eta},\enter{\eta},$ $\guard{\eta},\res{\eta},\linebreak\rf{\eta},\e{\eta})$,
where $\exit{\eta}$, $\enter{\eta}$, $\guard{\eta}$, $\e{\eta}$,
and $\res{\eta}$ are as for transitions in $\TD{}$.
\begin{itemize}
\item The \emph{rate} of the edge is 
$\rf{\eta}:\bbR^n\rightarrow \bbR^+$, a rate
function giving the instantaneous probability of taking transition
$\eta$. It is locally \emph{integrable} along any continuous differentiable
curve.  Additionally, transitions labelled by the same event are
required to have 
\emph{consistent} rates: if $\e{\eta_1} = \e{\eta_2}$, then
$\rf{\eta_1}=\rf{\eta_2}$.  
\end{itemize}

\item $\events{}$ is a finite set of event names, labelling discrete
transitions.  $\events{}$ can be partitioned into
$\events{d}\cup\events{s}$, such that all events labelling instantaneous
transitions belong to $\events{d}$, while all events labelling
stochastic transitions are from $\events{s}$.  

\item $\initial{}$ is a pair $(\initmode{},\mathbf{W})$, with
$\initmode{}\in\modes{}$ and $\mathbf{W}$ is a random vector of $n$
variables representing a point in
$\bbR^n$ and providing the initial values for $\vrX{}$.
$\initial{}$ describes the initial state of the system.

\end{itemize}
\end{definition}

Note that in the previous definition, we consider sets of instantaneous
transitions and multisets of stochastic and continuous transitions.
This is justified by the fact that multiplicity plays a relevant
role only in quantified behaviours (i.e. flows and stochastic
events), but it is not really relevant for instantaneous events,
which are triggered by a boolean condition, whose truth is insensitive
to the presence of multiple transitions of the same kind.

In order to formally define the dynamical evolution of TDSHA, we
can map them into a well-studied model of Stochastic Hybrid Automata,
namely Piecewise Deterministic Markov Processes~\cite{Davis93}.
Here we sketch basic ideas about the dynamical behaviour of TDSHA.
\begin{itemize}
\item Within each discrete mode $\mode{}\in \modes{}$, the system
follows the solution of a set of ODEs, constructed combining the
effects of the continuous transitions $\tau$ acting on mode $q$.
The function $\rf{\tau}(\vr{X})$ is multiplied by the vector
$\stoich{\tau}$ to determine its effect on each variable and then
all such functions are added together, so that the ODEs in mode $q$
are $\dot{\vr{X}} = \sum_{\tau~|~\mode{\tau} = q}  \stoich{\tau}\cdot
\rf{\tau}(\vr{X})$.

\item Two kinds of discrete jumps are possible. Stochastic transitions
are fired according to their rate, similarly to standard Markovian
Jump Processes. Instantaneous transitions, by contrast, are fired
as soon as their guard becomes true.  In both cases, the state of
the system is reset according to the specified reset policy which
is not necessarily deterministic.  Choice among several active
stochastic or instantaneous transitions is performed probabilistically
proportionally to their rate or weight.

\item A trace of the system is therefore a sequence of instantaneous
and random jumps interleaved by periods of continuous evolution.
\end{itemize}

Previously \cite{BortBGH:11a}, we defined a synchronization product
between TDSHA which is presented in Appendix~\ref{sec:SOSmapping} where it
is used to define the mapping used in our previous work on
stochastic HYPE \cite{BortBGH:11a}.

\section{Mapping stochastic HYPE to TDSHA}
\label{sec:HYPEtoTDSHA}

We present now a mapping from stochastic HYPE to TDSHA which is
based on the LTS obtained by the operational semantics. As mentioned
before, this construction is different from the one defined 
previously in \cite{BortBGH:11a}, in which we converted each subcomponent
and each controller of stochastic HYPE to TDHSA, and then combined these TDSHA
with a synchronization product. Essentially, the operational semantics
takes care of this synchronization, and avoids the complications
of the TDSHA product construction.

The mapping from LTS to TDSHA is quite straightforward: modes are
given by the derivative set of the controlled system $\ConSys$,
namely by the set of reachable configurations, continuous transitions
are extracted from the state in each configuration, and instantaneous
and stochastic transitions correspond to the LTS transitions.
Multiplicities of transitions labelled with stochastic events
have to be respected to properly
account for the quantitative behaviour of stochastic events, but as
mentioned previously, we ignore multiplicity of instantaneous
transitions, choosing to have a single transition with weight one.
%for the weights of instantaneous transitions. 
We assume that
multiplicity of elements in a multiset is always respected in the
definitions below.
 
Consider a stochastic HYPE model
$(\ConSys,\V,\IN,\IT,\Ev_d,\Ev_s,\Ac,\ec,\iv,\EC\!,\ID)$, and let
$\blang P_0, \sigma_0 \brang$ be the configuration in the LTS
obtained from $\ConSys$ after the execution of the initial event,
$\ev{init}$,
namely
$\blang \ConSys, \sigma \brang \mlsxrightarrow{\ev{init}}
\blang P_0, \sigma_0 \brang$. In the case of a well-defined
HYPE model $\sigma_0$ does not depend on $\sigma$, as shown by
\cite{HYPE-journal} and this applies to stochastic HYPE models as well.  We also let $\vr{v_0}$  be the value of
variables $\V$ after initialisation.

In the following, we will refer to the activation condition and the
reset of an event $\ssev{a}\in\Ev$ by $\act(\ssev{a})$ and
$\rs(\ssev{a})$, respectively; hence $\ec(\ssev{a}) =
(\act(\ssev{a}),\rs(\ssev{a}))$. Recall that, for an instantaneous
event $\ev{a}$, $\act(\ev{a})$ is  a boolean formula  while, for a
stochastic event $\sev{a}$, $\act(\sev{a})$ is a function taking
values in the positive reals.  Given a state $\sigma = \{
\iota_i\mapsto (r_i,I_i)~|~i=1,\ldots k\}$, we indicate with
$\vr{1}_{\iv(\iota_i)}$ the $k$-vector equal to 1 in the coordinate
corresponding to variable $\iv(\iota_i)$ and zero elsewhere.

The reset $\rs(\ssev{a})$ is defined as the conjunction of atoms
of the form $V \sim \theta(\V)$. To obtain a reset in the TDHSA,
each distribution $\calX_i$ mentioned in $\theta(\V)$ must be
associated with a distinct variable $W_i$, then the reset function for
a transition labelled with $\ssev{a}$ is $R_{\ssev{a}}(\V,\mathbf{W})$
such that each random variables $W_i$ is drawn from its associated
distribution.  

The mapping now defined differs from the approach we took for
standard HYPE, since it requires stochastic transitions.

\newpage

\begin{definition}
\label{def:shypetotdsha}
Let $\calH = (\ConSys,\V,\IN,\IT,\Ev_d,\Ev_s,\Ac,\ec,\iv,\EC\!,\ID)$
be a stochastic HYPE model. The TDSHA $\TDSHA{}(\calH) = \TDSHAtuple{}$
associated with $\calH$ is defined by:
\begin{itemize}
\item the set of discrete modes is the derivative set of $\ConSys$: 
$Q = ds(\ConSys)$;  
\item the set of continuous variables is $\vr{X} = \V$; 
\item the initial state is 
$\mathrm{init} = (\blang P_0, \sigma_0 \brang,R_{\ev{init}})$;
\item the set of events is $\mathcal{E} = \Ev_d\cup\Ev_s$;
\item the multiset of continuous transitions $\TC{}$ is constructed as
follows: \newline with each $\blang P, \sigma \brang \in Q$, $\sigma = \{
\iota_i\mapsto (r_i,I_i)~|~i=1,\ldots k\}$, we associate the
continuous transitions $(\blang P, \sigma
\brang,\vr{1}_{\iv(\iota_i)},r_i.\assign{I_i})$, $i=1,\ldots k$;
\item the set of instantaneous transitions $\TD{}$ is obtained as
follows: \newline with each $\blang P_1, \sigma_1 \brang \mlsxrightarrow{\ev{a}}
\blang P_2, \sigma_2 \brang$ we associate the instantaneous transition
\newline $(\blang P_1,
\sigma_1 \brang,\blang P_2, \sigma_2
\brang,\act(\ev{a}),R_{\ev{a}},p,\ev{a})$ with weight
$p=1$. We ignore multiplicity in
the LTS.
\item the multiset of stochastic transitions $\TS{}$ is constructed as
follows: \newline with each $\blang P_1, \sigma_1 \brang \mlsxrightarrow{\sev{a}}
\blang P_2, \sigma_2 \brang$ 
taking into account multiplicity in the LTS
we associate the stochastic transition \newline $(\blang P_1,
\sigma_1 \brang,\blang P_2, \sigma_2
\brang,true,R_{\sev{a}},\act(\sev{a}),\sev{a})$ with 
guard $\true$ and rate function $\act(\sev{a})$ and thus preserve
multiplicity.
\end{itemize}
\end{definition}

This definition gives the same TDSHA as that in the earlier work
on stochastic HYPE \cite{BortBGH:11a} which used a product construction
(the compositional mapping);
however, we prefer to work first with a labelled transition system
before mapping to a TDHSA as it allows us to define bisimulation
at the level of stochastic HYPE models. 
We can prove that the semantics defined above and those presented
in \cite{BortBGH:11a} map a stochastic HYPE model to the same TDSHA
and this theorem and its proof can be found in
Appendix~\ref{sec:SOSmapping}.

\begin{figure}
\begin{center}
\includegraphics[width=12.5cm]{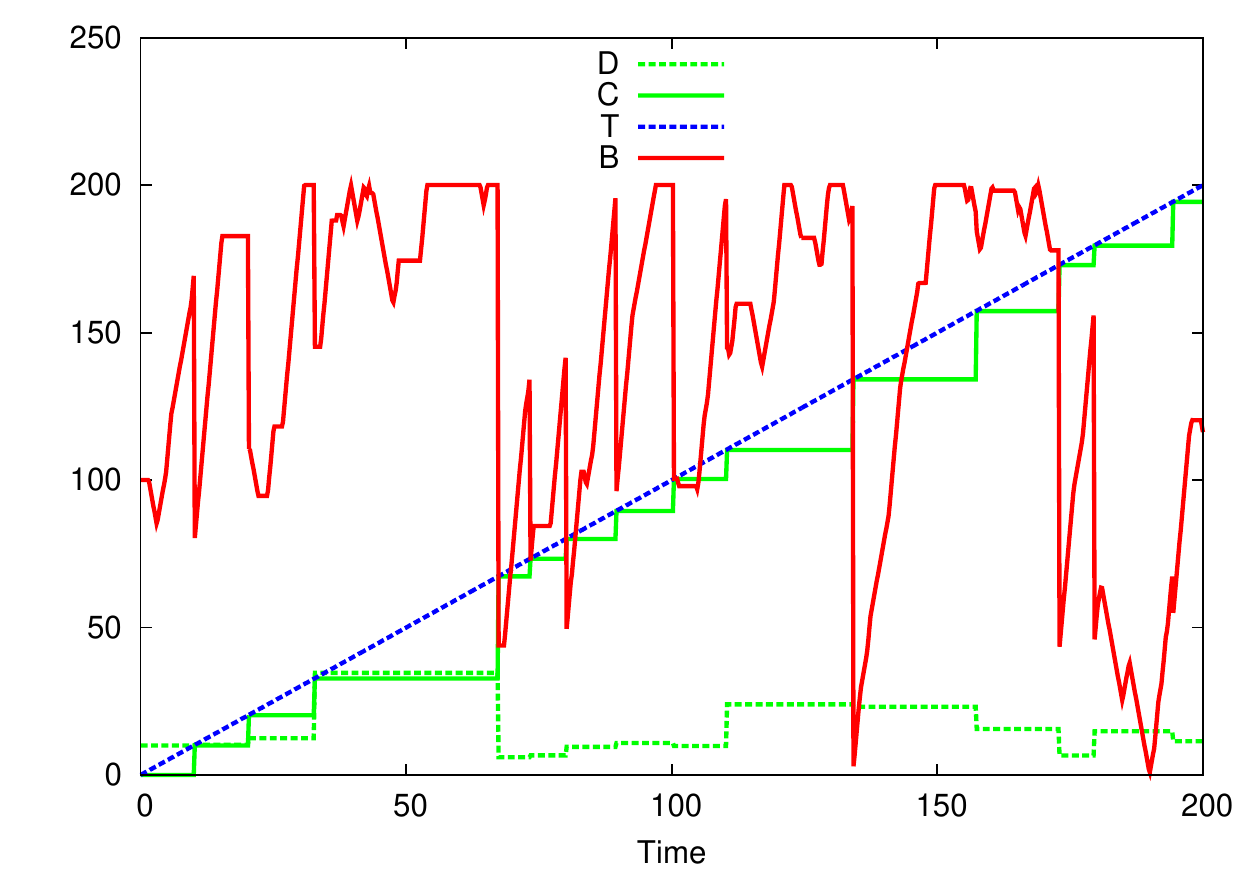}
\end{center}
\caption{A single trace for the network node with
$\mathit{max}_B=200$, $k_\mathit{in}^\mathit{on}=0.4$, 
$k_\mathit{in}^\mathit{off}=k_\mathit{out}^\mathit{off}=
k_\mathit{out}^\mathit{on}=0.2$, $r_\mathit{in}=20$, 
$r_\mathit{out}=10$, $\Delta=2.5$, $\xi=0.5$ and $b_0=100$.
The variable $B$ represents the quantity of packets in the buffer,
$T$ the time, $C$ the variable that records the time, and $D$ the
duration.}
\label{fig:node}
\end{figure}

A single trajectory of the node example from Section~\ref{sec:HYPE}
is given in Figure~\ref{fig:node}. The proportion of time that the
input connection is up and the rate of input are both higher than
that of the output node, and this can be seen in the fact that the
slopes for input  are steeper, and the buffer is hitting its maximum
capacity frequently. The failure that causes packets to be dropped
does, paradoxically, enable the buffer to accept more packets when
there is a connection.

We next consider some technical issues of model behaviour, as well a
syntactic abbreviation for events with non-exponential durations.
After this, the example of an assemble line shows the versatility of
stochastic HYPE and this is followed by the definition of an equivalence
over stochastic HYPE models that provides reasoning about when two
models have the same behaviour.

\subsection{Well-behaved models}

In the paper that defines HYPE \cite{HYPE-journal}, a notion of
well-behavedness is introduced. This notion ensures that the model
can never execute an infinite number of simultaneous instantaneous
events (called instantaneous Zeno behaviour), hence ensuring that
models can be simulated. We want stochastic HYPE models to be 
similarly well-behaved, so that models can be simulated and additionally
so that the TDHSA that a stochastic HYPE model is mapped to, can
be interpreted as piecewise deterministic Markov process (PDMP)
\cite{Davis93}.  The definition of PDMPs require that an instantaneous
event is not immediately followed by another one, and if the behaviour
of the stochastic HYPE model is such that in its TDHSA there are
only finite sequences of simultaneous events, then each finite
sequence can be mapped to a single event and thus satisfy the PDMP
definition. 

\begin{definition}
\label{def:well-behaved}
A stochastic HYPE model $P$ is well-behaved if it has a finite
number of finite sequences of simultaneous instantaneous events and
these sequences are independent of the initial state of the system.
\end{definition}

To ensure well-behavedness for a HYPE model (as defined in
\cite{HYPE-journal}, so without stochastic events), it is sufficient to
show that the instantaneous activation graph, or I-graph, of a HYPE
model is acyclic. This graph is constructed by considering the
instantaneous transitions of the labelled transition systems obtained
from a HYPE model. The I-graph gives an overapproximation of a HYPE
model's behaviour. If there are no cycles then the model is
well-behaved, but the I-graph of a well-behaved model is not necessarily
acyclic.

Considering stochastic HYPE models,
the addition of random resets has a minor effect on the construction
on the I-graph
in that for each event the possible set of values after a reset
must take into account the support of any distribution involved in
the reset, which can lead to additional overapproximation.  The
extension to stochastic events does not alter this construction and
all results from \cite{HYPE-journal} hold for stochastic HYPE models.
The addition of a stochastic event between two instantaneous events,
or alternatively, the modification of an instantaneous event to a
stochastic one may break the cyclicity of an I-graph thus leading
to a well-behaved model. Hence, both Theorem 6.1 and Theorem 6.2
proved by \cite{HYPE-journal} hold for stochastic HYPE, as well
as Propositions 6.1 to 6.4.  which describe specific conditions on
controllers that lead to well-behaved models. A new proposition can
be proved in the stochastic setting.

\begin{proposition}
\label{prop:well-behavedness}
Let $P$ be a stochastic HYPE model with $\Con \rmdef
\ssev{a}_1\ldots\ssev{a}_n.Con$. If there exists $i$ such that
$\sev{a}_i$ is a stochastic event then $P$ is well-behaved.
\end{proposition}
\begin{proof}
Since $\Con$ cycles through $n$ events, if one of these events is
not instantaneous and therefore has a duration then it is not possible 
for there to be an infinite sequence of instantaneous events.
\end{proof}

This proposition cannot be applied to the $\Buffer$ example, because
the controllers are not cyclical, but a similar argument can be
made. For the one instantaneous event in each controller ($\ev{full}$
or $\ev{empty}$), a stochastic event must occur before the instantaneous
event can reoccur, hence preventing the unwanted behaviour. Considering
the two controllers in cooperation, they have disjoint events and
hence the overall controller consists of interleavings of these
events. It is not possible for a sequence to occur consisting only
of $\ev{full}$ and $\ev{empty}$ since there must be interleaving
stochastic events. Hence $\Buffer$ is well-behaved. This can also
be proved using Proposition 6.4 from \cite{HYPE-journal} since
neither $\ev{full}$ nor $\ev{empty}$ activate each other, so one
cannot immediately proceed the other. Considering the controller
$\Con_\mathit{fail}$, the only way in which multiple $\ev{fail}$
events can happen simultaneously is if the random value for the
next duration is repeatedly zero. The probability of this is zero
and hence the controller is well-behaved. Proposition 6.4 can be
used again to show that the composition of all three controllers
is well-behaved.

\subsection{Non-exponential durations}

We also introduce some syntax that will allow us to write more compact
models. Currently the duration of stochastic events is specified by
exponential distributions because this is a good match with TDSHA and
PDMPs and more particularly because it avoids the need for residual
clocks. However, we can allow a notation whereby any expression involving
random variables can appear as the first element of an event condition.
This provides a way to express any random duration directly in an
event condition.  This will then be expanded to two events, and
requires the introduction of two variables, one to record the current
time and one to record the duration of the event. This introduces a
specific timer to track how much time is left of a duration.
To illustrate this, consider the following event.
\begin{eqnarray*}
\ec(\ev{fail}) & = & {(T=C\!+\!D,
\;\rsdefa{C}{T} \:\wedge\: \rsdefa{D}{\calN(\Delta,\xi)})
\:\wedge\: \rsdefa{B}{B\!-\!\calU(0,B)})}
\end{eqnarray*}
This can be written as 
\begin{eqnarray*}
\ec(\sev{fail}) & = & {(\calN(\Delta,\xi),
\rsdefa{B}{B\!-\!\calU(0,B)})}
\end{eqnarray*}
%which is then expanded to
%\begin{eqnarray*}
%\ec(\sev{fail}_s) & = & (\true,\rsdefa{C_{\ev{fail}}}{T} \:\wedge\:
%\rsdefa{D_{\ev{fail}}}{\calN(\Delta,\xi)}) \\
%\ec(\sev{fail}_f) & = & (T=C_{\ev{fail}}\!+\!D_{\ev{fail}},
%\rsdefa{B}{B\!-\!\calU(0,B)})
%\end{eqnarray*}
%This also requires the modification of any controllers where
%$\sev{fail}$ appears. Each occurrence of $\sev{fail}$ must be replaced
%with $\sev{fail}_s.\sev{fail}_f$. 
If there is no Timer
subcomponent with an influence affecting a time variable $T$, then these
must also be added. Since an expression that contains no random
variables is interpreted as the rate of an exponential distribution, the
notation $\delta(p)$ where $p$ contains no random variables will be used
to denote a fixed-time duration of $p$ time units.

We have now defined the dynamics of stochastic HYPE system as well
as highlight how models can have desirable behaviour. We assume
well-behavedness in the rest of the paper.  We proceed with an
example after which we consider how we can formally compare two
systems in terms of their behaviour.

\section{Example: a manufacturing system}
\label{sec:assemblyline}

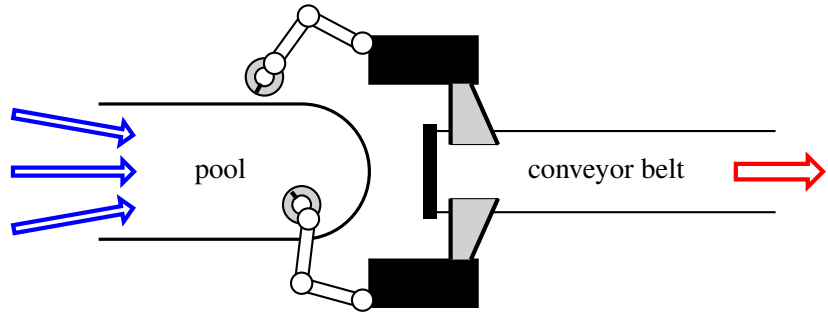
\begin{figure}
\vspace*{-0.3cm}
\begin{center}
%\begin{tikzpicture}[node distance=4cm,>=latex,thick,scale=0.65]
\begin{tikzpicture}[>=latex,thick,xscale=0.9,yscale=0.9,inner sep=0]
%\tikzstyle{every state}=[text height=0.4cm]

\path[very thick,draw] (0,1) -- (3,1) arc (90:0:1cm) arc (360:270:1cm);
\path[very thick,draw] (0,-1) -- (3,-1);

\node[] at (1.8,0) {pool};

\path[thick,draw] (10,0.6) -- (4.9,0.6) -- (4.9,-0.6) -- (10.0,-0.6);
\path[fill=black] (5.0,0.7) -- (5.0,-0.7) -- (4.8,-0.7) -- (4.8,0.7);
\node[] at (7.5,0) {conveyor belt};

\node[single arrow,draw,red,rotate=0,ultra thick,
      single arrow head extend=0.1cm]
     at (10,0.0) {\rl{1.0cm}{0.15cm}};

\node[single arrow,draw,blue,rotate=0,ultra thick,
      single arrow head extend=0.1cm]
     at (-0.4,0.0) {\rl{1.5cm}{0.05cm}};

\node[single arrow,draw,blue,rotate=-10,ultra thick,
      single arrow head extend=0.1cm]
     at (-0.4,0.7) {\rl{1.5cm}{0.05cm}};

\node[single arrow,draw,blue,rotate=10,ultra thick,
      single arrow head extend=0.1cm]
     at (-0.4,-0.7) {\rl{1.5cm}{0.05cm}};

\path[draw,fill=black] (4,1.3) -- (5.6,1.3) -- (5.6,2) -- (4,2) -- cycle;
\path[thin,draw,fill=black!18] (5.2,1.3) -- (5.2,0.4) -- (5.9,0.4) -- (5.5,1.3);
\path[ultra thick,draw] (5.2,1.3) -- (5.2,0.4);
\path[ultra thick,draw] (5.9,0.4) -- (5.5,1.3);

\path[draw,fill=black] (4,-1.3) -- (5.6,-1.3) -- (5.6,-2) -- (4,-2) -- cycle;
\path[thin,draw,fill=black!18] (5.2,-1.3) -- (5.2,-0.4) -- (5.9,-0.4) -- (5.5,-1.3);
\path[ultra thick,draw] (5.2,-1.3) -- (5.2,-0.4);
\path[ultra thick,draw] (5.9,-0.4) -- (5.5,-1.3);

\path[thick,draw,double distance=4pt] (3.85,1.9) -- (3.21,2.24);
\path[thick,draw,double distance=4pt] (3.05,2.24) -- (2.55,1.55);
\node[circle,draw,fill=white] at (3.9,1.9) {\rl{0.2cm}{0.2cm}};
\node[circle,draw,fill=white] at (3.1,2.3) {\rl{0.2cm}{0.2cm}};
\node[circle,draw,fill=black!18] at (2.45,1.4) {\rl{0.4cm}{0.30cm}};
\node[circle,draw,fill=white] at (2.45,1.4) {\rl{0.2cm}{0.15cm}};
\node[circle,draw,fill=white] at (2.6,1.6) {\rl{0.2cm}{0.2cm}};
\path[draw,ultra thick] (2.4,1.3) -- (2.33,1.17);

\path[thick,draw,double distance=4pt] (3.85,-1.9) -- (3.01,-1.67);
\path[thick,draw,double distance=4pt] (2.99,-1.67) -- (3.10,-0.75);
\node[circle,draw,fill=white] at (3.9,-1.9) {\rl{0.2cm}{0.2cm}};
\node[circle,draw,fill=white] at (3.0,-1.65) {\rl{0.2cm}{0.2cm}};
\node[circle,draw,fill=black!18] at (3.00,-0.49) {\rl{0.4cm}{0.30cm}};
\node[circle,draw,fill=white] at (3.00,-0.49) {\rl{0.2cm}{0.15cm}};
\node[circle,draw,fill=white] at (3.1,-0.7) {\rl{0.2cm}{0.2cm}};
\path[draw,ultra thick] (2.85,-0.30) -- (2.93,-0.40);

%\node[\rectangle] at (1

%\node at (-2.0,0.0) {};

%\node[right] at (0,-0.2) {\small membrane};
%\node[right] at (0,5.40) {\small perinuclear};
%\node[right] at (0,5.75) {\small region};

%\path[>=stealth,->,draw,orange,ultra thick] (-0.1,1.5) -- (-0.1,3.5);
%\path[>=stealth,->,draw,orange,ultra thick] (-0.3,3.5) -- (-0.3,1.5);
%\node[orange,right] at (-0.5,3.35) 
   {\small\begin{rotate}{90}10 seconds\end{rotate}};

%\path[>=stealth,->,draw,green!60!black,ultra thick] (10.3,1.5) -- (10.3,4.5);
%\path[>=stealth,->,draw,green!60!black,ultra thick] (10.5,4.5) -- (10.5,1.5);
%\node[green!60!black,left] at (10.7,2.4) 
   {\small\begin{rotate}{270}20 seconds\end{rotate}};

\end{tikzpicture}
\end{center}
\vspace*{-0.3cm}
\caption{Schematic of the assembly system}
\label{fig:assembler}
\end{figure}

\begin{figure}%[t]
{\renewcommand{\arraystretch}{1.10}
$$\begin{array}{rcl}
\\
\\
\Assem & \rmdef & Sys \syncstar \ev{init}.\Con_j \\
& & \\
Sys & \rmdef & (\Feed_1 \syncstar \Feed_2 \syncstar Feed_3) \syncstar 
               \Inspect \syncstar \\
    &        & (\Timer_1 \syncstar \Machine_1(W_1)) \syncstar \\
    &        & (\Timer_2 \syncstar \Machine_2(W_2)) \\
& & \\
\Feed_i    & \rmdef & \ev{init}\pc(p_i,\arr_i,\const).\Feed_i\:+\\
           &        & \ev{overflow}\pc(p_i,0,\const).\Feed_i\\
\\
\Inspect   & \rmdef & \ev{init}\pc(b,-\dep,\const).\Inspect\:+\\
           &        & \sev{scan}\pc(b,-\dep,\const).\Inspect \\
           &        & \sev{resume}\pc(b,0,\const).\Inspect \\
           &        & \ev{overflow}\pc(b,0,\const).\Inspect \\
\\
\Machine_i(W_i) & \rmdef &
    \ev{init}\pc(w_i,\wa_i,\linear(W_i)).\Machine_i(W_i)\:+\\
& & \ev{check}_i\pc(w_i,0,\const).\Machine_i(W_i)\:+\\
& & \sev{remove}_i\pc(w_i,\wt_i,\linear(W_i)).\Machine_i(W_i)\:+ \\
& & \ev{assem}_i\pc(w_i,\wa_i,\linear(W_i)).\Machine_i(W_i) \\
& & \ev{overflow}_i\pc(w_i,0,\linear(W_i)).\Machine_i(W_i) \\
\\
\Timer_i   & \rmdef & \ev{init}\pc(t_1,0,\const).\Timer_i\:+ \\
           &        & \ev{remove}_i\pc(t_1,1,\const).Timer_i\:+ \\
           &        & \ev{assem}_i\pc(t_1,0,\const).\Timer_i\\
\\
\iv(p_i)   & = & P \quad \iv(t_i) \: = \: T_i \quad \iv(w_i) \: = \: W_i
\quad \iv(b) \: = \: B \\
\\
\ec(\ev{init}) & = & (true,\;\rsdefa{P}{P_0} \:\wedge\: \rsdefa{T_1}{0}
\:\wedge\: \rsdefa{T_2}{0} \:\wedge\: \rsdefa{B}{B_0}) \\
\ec(\ev{overflow}) & = & (B \geq B_f,\;\true) \\
\ec(\ev{check}_i) & = & (P\geq n_i,\;\true) \\
\ec(\ev{assem}_i) & = & (T_i \geq \assemtime_i,\; \rsdefa{B}{B\!+\!m_i}) \\
\\
\ec(\sev{remove}_i) & = & (\taketime_i,\;\rsdefa{P}{P\!-\!n_i} \:\wedge\:
\rsdefa{T_i}{0}) \\
\ec(\sev{scan}) & = &
(\examtime,\;\rsdefa{B}{B-min(B,\Gamma(S_c,S_h))}) \\
\ec(\sev{resume}) & = & (\calF(\resttime),\;\true) \\
\\
\eval{const} & = & 1 \qquad \eval{linear(X)} \: = \: X \\
\end{array}$$
}
\caption{Model for assembler with two machines (controller omitted)}
\label{fig:assemuncontrolled}
\end{figure}

\begin{figure}
{\renewcommand{\arraystretch}{1.20}
$$\begin{array}{rcl}
\\
\Con & \rmdef & ((C_1 \parallel C_2) \syncstar C_m) \parallel 
C_e \parallel C_f \\
\\
C_i & \rmdef & \ev{check}_i.C'_i \\
C'_i & \rmdef & \sev{remove}_i.C''_i \\
C''_i & \rmdef & \ev{assem}_i.C_i \\
\\
C_m & \rmdef & \ev{check}_1.C'_{m} +
               \ev{check}_2.C''_{m} \\
C'_m & \rmdef & \sev{remove}_1.C_m \\
C''_m & \rmdef & \sev{remove}_2.C_m \\
\\
C_e & \rmdef & \sev{scan}.\sev{resume}.C_e \\
\\
C_f & \rmdef & \ev{overflow}.0 
%\Con_2 & \rmdef & \ABOff \parallel \FullCheck \\
%\\
%\ABOff   & = & \sev{prep}.\AIOn     + \sev{prep}.\AIOn; \\
%\AIOn    & = & \sev{prep}.\ABOn     + \sev{prep}.ABOn + \\
%         &   & \ev{take}_1.\APrI    + \ev{take}_2.\APrZ; \\
%\ABOn    & = & \ev{take}_1.\APrIOnZ + \ev{take}_2.\APrZOnI\\
%\ABPr    & = & \ev{assem}_1.\APrZ   + \ev{assem}_2.\APrI \\
%\APrI    & = & \ev{assem}_1.\ABOff  + \sev{prep}.\APrIOnZ \\
%\APrZ    & = & \ev{assem}_2.\ABOff  + \sev{prep}.\APrZOnI \\
%\APrIOnZ & = & \ev{assem}_1.\AIOn   + \ev{take}_2.\ABPr\\
%\APrZOnI & = & \ev{assem}_2.\AIOn   + \ev{take}_1.\ABPr\\
%\FullCheck & \rmdef & \ev{full}.0 \\
\end{array}$$}

\caption{Controller for the assembly system}
\label{fig:assemcontroller}
\end{figure}

The example considers automated machines for assembling together
groups of identical items, for example, putting matches into
matchboxes. A schematic of the system is presented in
Figure~\ref{fig:assembler}. Each machine determines if there are
sufficient items in the pool ($\ev{check}_i$), then it takes these
$n_i$ items. Since this action can vary in duration it is modelled
as a stochastic event ($\sev{remove}_i$) with a exponentially
distributed duration.  The next step is assembly and the event
$\ev{assem}_i$ indicates the end of this fixed duration process
when the machine places $m_i$ finished items onto an output conveyor
belt. Only one machine is allowed to take items from the pool at a
time and this is enforced by a controller.

There are three lines of items that feed into the pool. Completed
items are removed from the conveyer belt. At the conveyer belt, there
is an agent that stops the belt, inspects the items near to it,
removes incorrectly assembled items and restarts the belt. If the
beginning of the conveyer belt becomes congested and starts to
overflow, the system moves to a failsafe state where everything
stops.

Here, we treat both the input feed into the pool ($P$) and the output
belt ($B$) as continuous.
We also track the power consumption of each machine ($W_i$).
The uncontrolled system is defined in Figure~\ref{fig:assemuncontrolled}
and its controller are presented in Figure~\ref{fig:assemcontroller}.
$\Con$ consists of five controllers: one for each
machine, a controller that controls access to the pool,
a controller for inspection of assembled items, and 
a controller that determines the congestedness of the
output belt and shuts down the whole system if necessary.  
Figure~\ref{fig:parrun} provides a single run of the system.

We now consider a more complex scenario to show the potential of
stochastic HYPE. This example involves the following cost optimisation
problem. The manufacturer receives an order for $K$ products that
have to be produced within a certain deadline.  Failure to meet the
deadline will result in a penalty proportional to the delay. We
assume that the assembly machines can be tuned by changing the batch
size of a single assembly: the machine can take more items from the
pool and put more assembled items on the belt. This has a cost in
terms of assembly time and energy consumption. However, we assume
that the production time increases as the square root of the batch
size, while the energy increases quadratically.  This models the
fact that increasing the batch size reduces the production time but
at an increased energy cost. Energy itself contributes to the total
cost at a given price per unit. The goal of the manufacturing is
to find the batch size that minimises the average cost, defined as
the energy cost plus the penalty to miss the deadline. This design
problem can be solved by exploring parameters in a feasible range
and looking for the minimum value of the average cost, and results
are reported in Figure ~\ref{fig:intrun}  (taking the average over
1000 runs per point), suggesting an optimal batch size of 2. The model
for this optimisation requires minor changes to the model in
Figure~\ref{fig:assemuncontrolled} and these are easy to make because
of the structured form of the model. Assembly times and the flows
describing energy consumption are modified, and for the purposes
of the example, the area where the manufactured items are placed
is made much larger because the focus is on production, rather than
problems related to overflow. No changes are required for the
controllers.

\begin{figure}
\begin{center}
\begin{tabular}{cc}
\includegraphics[width=12.5cm]{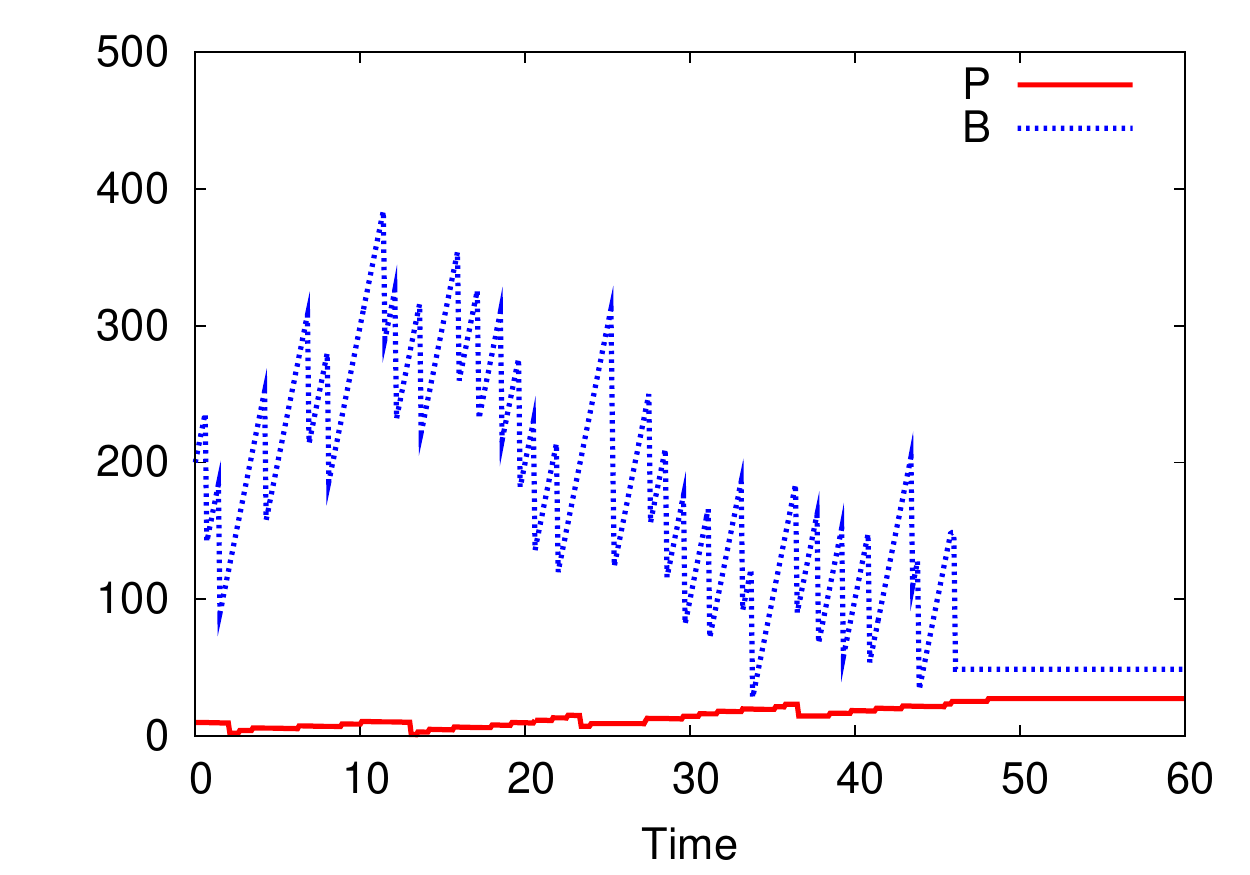} \\
\includegraphics[width=12.5cm]{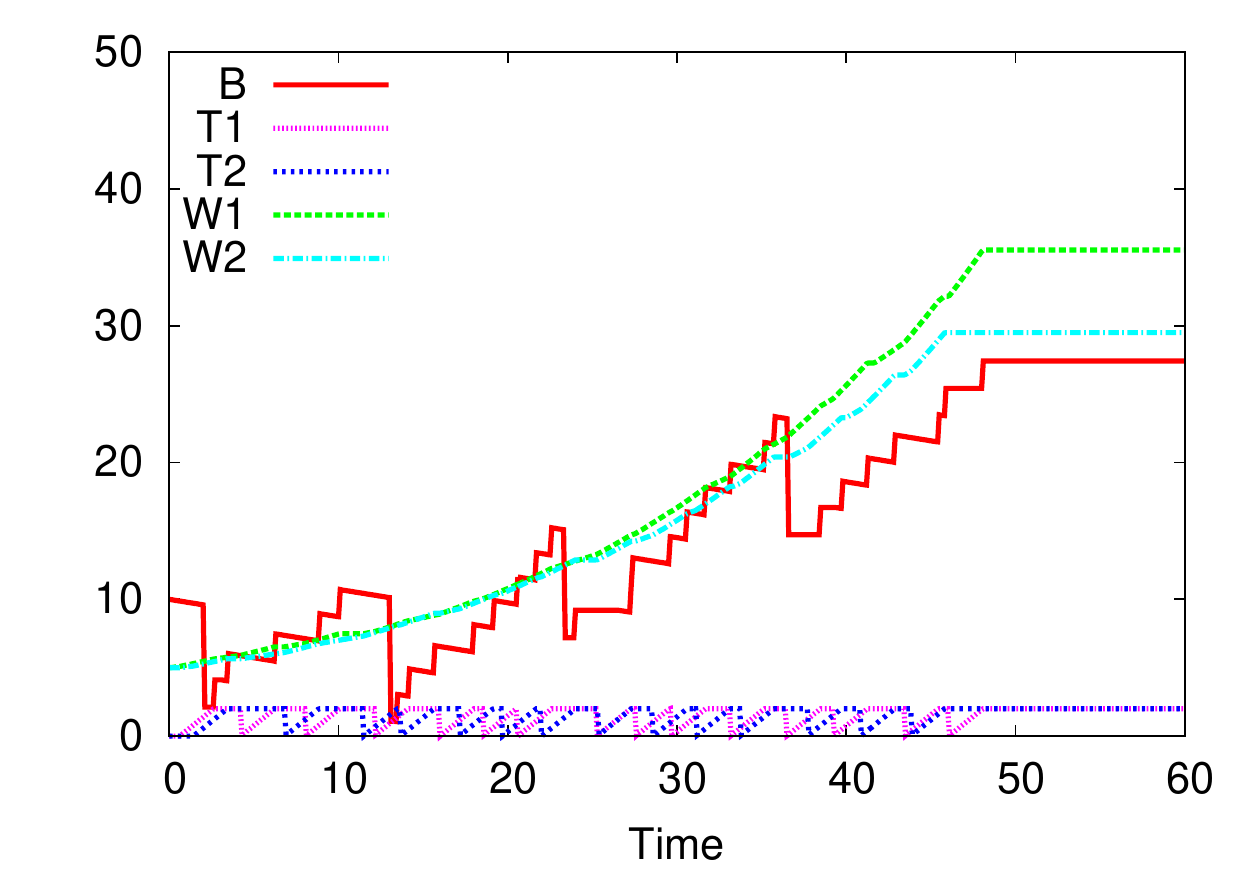}
\end{tabular}
\end{center}
\caption{A single trace for the assembly system showing
variables $P$ and $B$ (top) and variables $B$, $T_i$ and $W_i$
(bottom) using the parameters $\arr_i=20$, $\dep=0.2$, $\assemtime_i=2$,
$\taketime=0.8$, $n_i=100$, $m_i=2$, $\wt_i=0.03$,
$\wa_i=0.05$, $\examtime=2$, $\resttime=20$, $S_c=4$ $S_h=0.5$ and
$B_f = 25$}
\label{fig:parrun}
\end{figure}

\begin{figure}
\begin{center}
\begin{tabular}{cc}
\includegraphics[width=12.5cm]{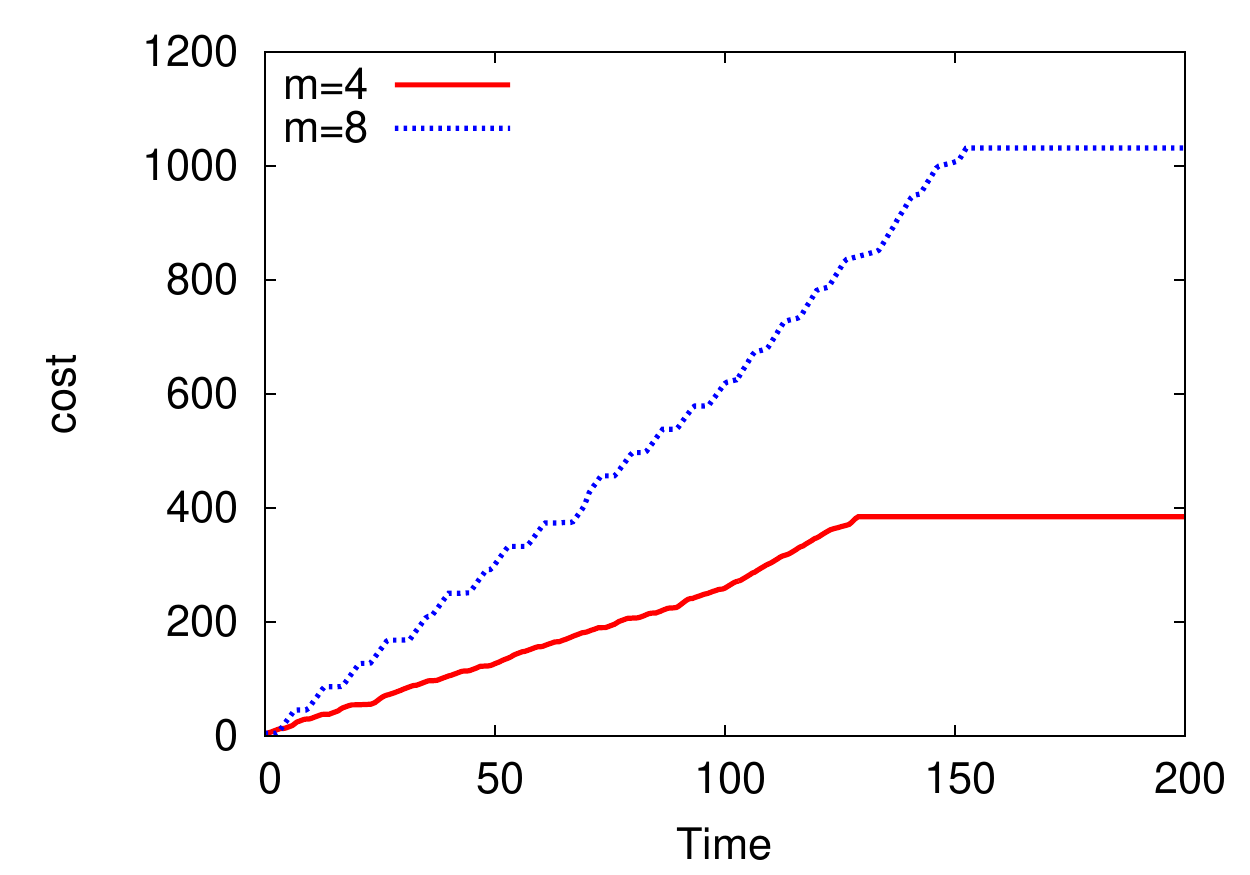} \\
\includegraphics[width=12.5cm]{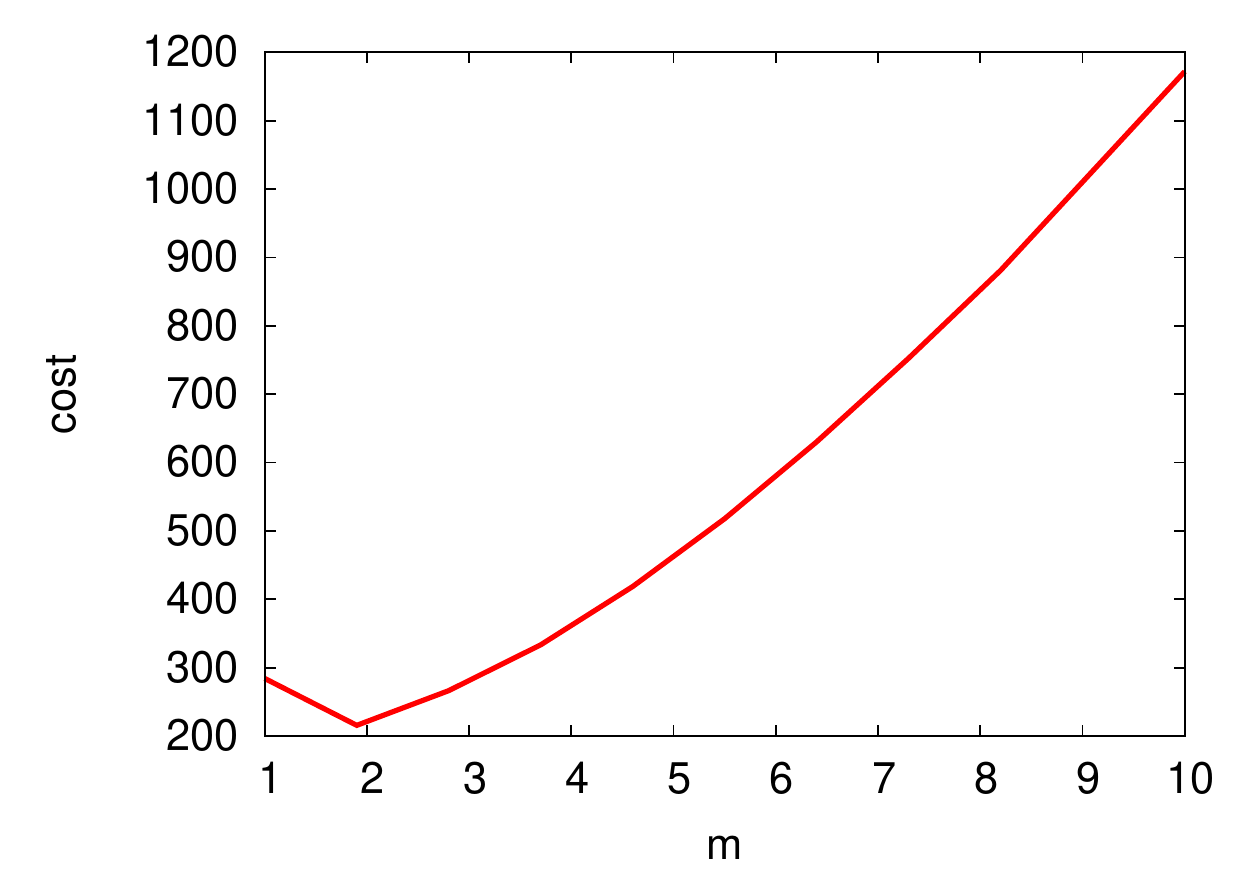}
\end{tabular}
\end{center}
\caption{
Energy plus penalty cost for various values of batch size $m$,
assuming both machines are configured to the same value. Two
trajectories for different values of $m$ (top) and average cost for a range of
values for $m$ over 1000 trajectories (bottom).
Energy cost per
unit of energy is 0.5, while penalty cost is 2 per time unit.  Energy
rate consumption in the assembly phase is equal to $1/3*m^2 + 2/3$,
while the duration of the assembly equals $\alpha*(\sqrt{m}-1) + 1.5$,
for $\alpha = 1/(2*sqrt(2) - 2)$. The number of items taken from the
pool is $50 m$.  We assume the order is of 100 items with a
deadline of 100 time units.}
\label{fig:intrun}
\end{figure}

\section{Bisimulations for stochastic HYPE}
\label{sec:bisimHYPE}

We now consider behavioural equivalences for well-behaved stochastic
HYPE models. First, we
show that how the natural extension of the bisimulation defined for HYPE
is not a useful definition when stochastic hybrid models are considered.
Next, we define an equivalence that is more
suited to capture the notion of same behaviour both stochastically
and instantaneously.

\subsection{System bisimulation}

We have previously defined system bisimulation for HYPE \cite{HYPE-journal}
and shown that it is the same as ic-bisimulation
\cite{CuijCR:03a,BergBM:05a}. The definition of system bisimulation
along these lines for stochastic HYPE only requires the modification
of the labels on the transitions so that they can either be stochastic
or instantaneous.

\begin{definition}
A relation $B \subseteq \calC \times \calC$
is a \emph{system bisimulation} if for all $(P,Q) \in B$, for all
$\ssev{a}
\in \Ev_d \cup \Ev_s$, for all $\sigma \in \calS$ whenever
\begin{enumerate}
\item $\cf{P,\sigma} \xrightarrow{\ssev{a}} \cf{P',\sigma'}$, there exists
$\cf{Q',\sigma'}$ with $\cf{Q,\sigma} \xrightarrow{\ssev{a}}
\cf{Q',\sigma'}$,
$(P'\!,Q') \in B$.
\item $\cf{Q,\sigma} \xrightarrow{\ssev{a}} \cf{Q',\sigma'}$, there exists
$\cf{P',\sigma'}$ with $\cf{P,\sigma} \xrightarrow{\ssev{a}}
\cf{P'\!,\sigma'}$,
$(P',Q') \in B$.
\end{enumerate}
$P$ and $Q$ are \emph{system bisimilar},
$P \sim_s Q$ if they are in a system
bisimulation.
\end{definition}

The results that held for HYPE also hold for stochastic HYPE, as
dealing with stochastic events in the proofs is straightforward.
Hence, system bisimulation is a congruence for Prefix, Choice and
Parallel, and if two uncontrolled systems have the
same set of prefixes, and are put in cooperation with the same
controller, then the two controlled systems are bisimilar. System
bisimulation can be lifted to the model level and congruence can
be shown for stochastic HYPE model product.  
%Additionally, it can be shown that system bisimilar models generate the 
%same ODEs at each mode in their respective TDSHAs.  
As these are all straightforward
modifications of existing proofs, we omit presenting them formally
for reasons of space. The reader is referred to \cite{HYPE-journal}
for further details.

System bisimulation is a static bisimulation in the sense that it
does not consider the detailed behaviour of the model. It considers
which events can occur, matches on them and also matches strictly
on state. However, for stochastic events it requires exact matching
of rates, rather than of the overall rate to processes with the
same behaviour.  In the next section, we consider a less strict and
more useful equivalence.

\subsection{An equivalence based on stochastic bisimulation}

We now consider a bisimulation that takes conditions on dynamics
into account and relaxes the strict matching on states.  Furthermore,
we wish to relax the requirements for matching of stochastic
transitions, and move to a definition that is similar to stochastic
bisimulation \cite{Hill96,BuchB95a} where the combined rate to an equivalence
class for each action is considered. For this, we require additional
definitions. The sum in the definition is taken over a multiset.

\begin{definition}
For the transitions $\blang P_1, \sigma_1 \brang
\lsxrightarrow{\:\ssev{a}\:}
\blang P_2, \sigma_2 \brang$ let $\multi(\blang P_1, \sigma_1
\brang \lsxrightarrow{\:\ssev{a}\:} \blang P_2, \sigma_2 \brang)$ be the
number of such transitions.
\end{definition}

\begin{definition}
Given a stochastic HYPE controlled system $P = \Sigma \syncstar \Con$,
the function $r: \calF \times \Ev_s \times \calC \rightarrow \bbR^+$ is defined
as
\[ r(\cf{P,\sigma},\sev{a},\Sigma \syncstar \Con') = \begin{cases}
      \act(\sev{a})\cdot\multi(\blang P, \sigma
\brang \mlsxrightarrow{\sev{a}} \blang \Sigma \syncstar \Con', \sigma' \brang)
 & \text{if\ } \Con \xrightarrow{\sev{a}} \Con' \\
      0             & \text{otherwise}
      \end{cases}
\]
Furthermore, this function can be extended to sets $C \subseteq \calC$,
$r(\cf{P,\sigma},\sev{a},C) = \sum \{ r(\cf{P,\sigma},\sev{a},Q) \mid Q
\in C \}$.
\end{definition}

\begin{definition}
Given $\equiv$, an equivalence relation over states,
an equivalence relation $B \subseteq \calC \times \calC$
is a \emph{stochastic system bisimulation with respect to $\equiv$} if 
for all $(P,Q) \in B$, for all
$\sigma \equiv \tau$, for all $C \in (\CF/B)/\equiv$,
\begin{enumerate}
\item for all $\ev{a} \in \Ev_d$, 
whenever
\begin{enumerate}
\item $\cf{P,\sigma} \xrightarrow{\ev{a}} \cf{P',\sigma'} \in C$, there exists
$\cf{Q',\tau'} \in C$ such that $\cf{Q,\tau} \xrightarrow{\ev{a}}
\cf{Q',\tau'}$. 
\item $\cf{Q,\tau} \xrightarrow{\ev{a}} \cf{Q',\tau'} \in C$, there exists
$\cf{P',\sigma'} \in C$ such that $\cf{P,\sigma} \xrightarrow{\ev{a}}
\cf{P'\!,\sigma'}$.
\end{enumerate}
\item for all $\sev{a} \in \Ev_s$, 
$r(\cf{P,\sigma},\sev{a},C) = r(\cf{Q,\tau},\sev{a},C)$. 
\end{enumerate}
$P$ and $Q$ are \emph{stochastic system bisimilar with respect to $\equiv$},
$P \sim_s^\equiv Q$ if they are in a stochastic system
bisimulation with respect to $\equiv$.
\end{definition}

Note that even in the case that $\equiv$ is equality, stochastic system
bisimulation with respect to $\equiv$ is less
strict than isomorphism over the instantaneous transitions.  Next,
we consider under which conditions $\sim_s^\equiv$ is a congruence.
We need two definitions to characterise the interaction of $\equiv$
and functions used in the operational semantics.

\begin{definition}
An equivalence $\equiv$ over states \emph{is preserved by
updates} if $\sigma \equiv \tau$ implies that
$\sigma[\iota \mapsto (r,I)]\equiv\tau[\iota \mapsto (r,I)]$.
\end{definition}

\begin{definition}
An equivalence relation $\equiv$ over states \emph{is preserved by $\Gamma$}
if $\sigma_i \equiv \tau_i$ for $i=1,2,3$ implies that
$\Gamma(\sigma_1,\sigma_2,\sigma_3) \equiv
\Gamma(\tau_1,\tau_2,\tau_3)$.  
\end{definition}

\begin{theorem}
\label{congruence}
$\sim_s^\equiv$ is a congruence for Influence, Choice and 
Cooperation, if $\equiv$ is preserved by updates and $\Gamma$.
\end{theorem}
\begin{proof} Please refer to Appendix~\ref{sec:congruence}.
\end{proof}

This theorem describes the conditions on the equivalence over states
for congruence. Both preservation by updates and by $\Gamma$ are
strong conditions, but as we will see later in this section, not
always necessary.

Because of the specific form of well-defined stochastic HYPE models,
congruence with respect to some operators is less important.  It
is not possible to obtain a well-defined stochastic HYPE model by
applying Prefix with Influence or Choice to an existing well-defined
stochastic HYPE model. However Cooperation and Prefix are used to
construct a controlled system from an uncontrolled system and
controller, hence congruence can be used to prove further results
as in the next theorem.

\begin{theorem}
\label{thm:equivcon}
Given an uncontrolled system $\Sigma$ and two controllers
such that $\Con_1 \sim_s^{\equiv} \Con_2$ and let $\equiv$
be preserved by $\Gamma$ then $\Sigma \syncstar \ev{init}.Con_1
\sim_s^{\equiv} \Sigma \syncstar \ev{init}.Con_2$ \end{theorem}
\begin{proof} By congruence.  \end{proof}

Next, we introduce a specific equivalence over states, and prove
that it gives the same ODEs for models that are stochastically
system bisimilar with respect to it.

\begin{definition} Two states are equivalent, $\sigma \doteq \tau$,
when for all $V \in \V$ and $f(\W)$, $\ad(\sigma,V,f(\W)) =
\ad(\tau,V,f(\W))$ where \[\ad(\sigma,V,f(\W)) = \sum \lms r \mid
\iv(\iota)=V, \sigma(\iota)=(r,I(\W)), f(\W)=\assign{I(\W)}\rms\]
\end{definition}

This equivalence abstracts from individual influences by requiring
that the sum of strengths for each variable and influence type is
preserved. It is not preserved by updates or $\Gamma$; however, as
will be seen in the example section, it still provides a useful
equivalence.  This is because, for certain models, it is the case
that $\sigma \doteq \tau$ implies $\Gamma(\sigma,\sigma',\sigma'')
\doteq \Gamma(\tau,\tau',\tau'')$ even though $\sigma' \not\doteq
\tau'$ and $\sigma'' \not\doteq \tau''$. This can be achieved by
imposing additive conditions on the rates in a model for specific
events.  To illustrate this, consider the subcomponents
\begin{eqnarray*}
A_i & \rmdef & \ev{a}\pc(k_i,r_i,I).A_i + \ev{init}\pc(k_i,0,I).A_i \\
B_i & \rmdef & \ev{a}\pc(j_i,s_i,I).B_i + \ev{init}\pc(j_i,t,I).B_i
\end{eqnarray*}
with $\iv(k_i) = X = \iv(j_i)$. After the $\ev{init}$ event, we
have states $\sigma_i = \{ k_i \mapsto (0,I), j_i \mapsto (t,I)\}$
and therefore $\ad(\sigma_i,X,I) = t$. Clearly, $\sigma_1 \doteq \sigma_2$.
However, after an $\ev{a}$ event, we have states $\tau_i = \{ k_i
\mapsto (r_i,I), j_i \mapsto (s_i,I)\}$ and $\ad(\tau_i,A_i \syncstar
B_i,I) = r_i + s_i$.  Hence $\tau_1 \doteq \tau_2$ if and only if
$r_1 + s_1 = r_2 + s_2$. This does not require that $r_1 = r_2$ and
$s_1 = s_2$ which are the conditions required for equivalent states
for $A_1$ and $B_1$, and $A_2$ and $B_2$.

Next, we wish to lift stochastic system bisimulation with respect to an
equivalence,  from controlled system level to model level, both to
consider congruence of model product and to impose conditions on
the elements of the tuples.

\begin{definition}
\label{def:systemBisim}
Let $(P_i,\,\V,\,\IN,\,\IT,\Ev_d,\Ev_s,\,\Ac,
\ec,\,\iv,\,\EC,\,\ID)$ for $i=1,2$ 
be two stochastic HYPE models. They are \emph{stochastic system bisimilar
with respect to an equivalance $\equiv$ over states}
(denoted $P_1 \thicksim^{\equiv}_\mathbf{sm} P_2$) if $P_1 \sim^{\equiv}_s P_2$.
\end{definition}

Let the notation $P_\sigma$ denote the collection of ODEs for model
$P$ in state $\sigma$. The next results shows that models that are
stochastic system bisimilar with respect to $\doteq$ have the same ODEs.

\begin{theorem}
\label{thm:doteqODES}
Given two stochastic HYPE models
$(P,\,\V,\,\IN,\,\IT,\Ev_d,\Ev_s,\,\Ac,
\ec,\,\iv,$ $\EC,\,\ID)$ and $(Q,\,\V,\,\IN,\linebreak\IT,\Ev_d,\Ev_s,\,\Ac,
\ec,\,\iv,\,\EC,\,\ID)$, if
$P \thicksim_\mathbf{sm}^{\doteq} Q$, then
for every configuration $\cf{P',\sigma_1} \in \ds(P)$ and
$\cf{Q',\sigma_2} \in \ds(Q)$ such that $P' \thicksim_\mathbf{sm}^{\doteq} 
Q'$, $P_{\sigma_1} = Q_{\sigma_2}$.
\end{theorem}
\begin{proof} Please refer to Appendix~\ref{sec:equiv}.
\end{proof}

We can also define bisimulation at the TDSHA level and relate the
stochastic HYPE bisimulation to that of the bisimulation over TDSHA.
In Appendix \ref{sec:bisimforTDSHA}, we show that two stochastic HYPE models
that are stochastic system bisimilar with respect to ${\doteq}$
have TDSHAs that are bisimilar but the converse does not necessarily
hold.

%\vspace*{1cm}
%\noindent
%==================================================================
%
%Stochastic HYPE models with same prefixes and controllers are
%equivalent.
%
%\begin{theorem}
%\label{prefbisim}
%Let $\Sigma_1 \smash{\sync{L}} \ssev{init}.\Con$ and $\Sigma_2
%\smash{\sync{L}}
%\ssev{init}.\Con$ be two controlled systems.
%If $\prs(\Sigma_1) = \prs(\Sigma_2)$ then $\Sigma_1
%\smash{\sync{L}}
%\ssev{init}.\Con \sim_s \Sigma_2 \smash{\sync{L}}
%\ssev{init}.\Con$.
%\end{theorem}
%
%\noindent
%==================================================================
%\vspace*{1cm}

\section{Example revisited: equivalence}

\begin{figure}[t]
{\renewcommand{\arraystretch}{1.20}
$$\begin{array}{rcl}
\Con_D & \rmdef & D \parallel C_e \parallel C_f \\
\\
D       & \rmdef & \ev{check}_1.D_{1,1} + \ev{check}_2.D_{1,2} \\
D_{1,i} & \rmdef & \sev{remove}_i.D_{2,i} \\
D_{2,i} & \rmdef & \ev{assem}_i.D + \ev{check}_i.D_{3,i} \\
D_{3,i} & \rmdef & \ev{assem}_i.D_{1,i+1} + \sev{remove}_i.D_{4} \\
D_{4}   & \rmdef & \ev{assem}_1.D_{2,2} + \ev{assem}_2.D_{2,1} \\
\end{array}$$}

\caption{Controller for the assembly system with a single controller
for the two assembly machines (addition is modulo 2)}
\label{fig:Dcontroller}
\end{figure}

We can now consider equivalence in the context of the manufacturing
system. We can define
$\Con_D$, a composite controller for two machines and access to the pool
as given in Figure \ref{fig:Dcontroller}, to replace the controllers $C_1$,
$C_2$ and $C_m$.

First, we show that ${\Sys \syncstar \ev{init}.\Con} \sim_s^{\equiv}
{\Sys \syncstar \ev{init}.\Con_D}$ for a suitable $\equiv$. This
can be done by Theorem \ref{thm:equivcon} and requires that a
suitable equivalence relation be identified over the two controllers.
We can ignore states inf configurations since they are not affected
by the events in the controller. The relation $B$ that follows is
a stochastic system bisimulation with respect to $=$, and illustrates that
therefore $(C_1 \parallel C_2) \syncstar C_m)
\sim_s^{=} D$ and by congruence $\Con \sim_s^{=} Con_D$,
since equality preserves $\Gamma$.
Hence ${\Sys \syncstar
\ev{init}.\Con} \sim_s^{=} {\Sys \syncstar \ev{init}.\Con_D}$.
\begin{eqnarray*}
B & = & \bigl\{ 
((C_1 \parallel C_2) \syncstar C_m, D), 
    ((C''_1 \parallel C''_2) \syncstar C_m, D_{4}), \\
& & ((C'_1 \parallel C_2) \syncstar C'_m, D_{1,1}), 
    ((C_1 \parallel C'_2) \syncstar C''_m, D_{1,2}), \\
& & ((C''_1 \parallel C_2) \syncstar C_m, D_{2,1}), 
    ((C_1 \parallel C''_2) \syncstar C_m, D_{2,2}), \\
& & ((C''_1 \parallel C'_2) \syncstar C_m, D_{3,1}), 
    ((C'_1 \parallel C''_2) \syncstar C_m, D_{3,2})
\bigr\}
\end{eqnarray*}
Figure~\ref{fig:SimCon2} shows the average for these two assembly systems
over 5000 simulations. The similarity between these averages suggest that
our definition of bisimulation has captured the similarities between the
two systems.

We note here that the stochastic HYPE system $\Assem$ (which uses the
five
subcontrollers) is well-behaved,
according to Definition \ref{def:well-behaved}. This holds because
each component of the controller is of the form required by Proposition
\ref{prop:well-behavedness}, and contains a stochastic event
($\sev{remove}_i$). To show that $\Assem_D$ (which uses three
subcontrollers one of which is 
$\Con_D$) is well-behaved, we can
simply invoke bisimilarity of the two systems, noticing that
well-behavedness, being a condition on sequences of events, is
preserved by bisimulation.

As mentioned above, the subcontrollers $(C_1 \parallel C_2) \syncstar C_m$
and $D$ both ensure that only one machine has access to the pool at a
time. Another approach is to use only the controller $(C_1 \parallel
C_2)$ and modify the event conditions for each machine. We add a new
variable $M$ and redefine some event conditions.

\begin{eqnarray*}
\ec(\ev{init}) & = & (true,\;\rsdefa{P}{P_0} \:\wedge\: \rsdefa{T_1}{0}
\:\wedge\: \rsdefa{T_2}{0} \:\wedge\: \rsdefa{B}{B_0} \:\wedge\:
\rsdefa{M}{0}) \\
\ec(\ev{check}_i) & = & (P\geq n_i \:\wedge\: M=0,\; \rsdefa{M}{1}) \\
\ec(\sev{remove}_i) & = & (\taketime_i,\;\rsdefa{P}{P\!-\!n_i} \:\wedge\:
\rsdefa{T_i}{0} \:\wedge\: \rsdefa{M}{0}) \\
\ec(\ev{assem}_i) & = & (T_i \geq \assemtime_i,\; \rsdefa{B}{B\!+\!m_i}) 
\end{eqnarray*}
This modifies the system from one where mutual exclusion is determined
by explicit sequencing of actions in the controller to one where a
semaphore is used. Since the definition of stochastic system
bisimulation has a requirement for events conditions to be equal for
each event, we can no longer directly use this. Instead we can reason
about the behaviour of the controllers from each system and show that
the beahviour are the same, taking into account the different event
conditions. We can then use Theorem~\ref{thm:equivcon} to argue for
the stochastic system bisimilarity of the two controlled systems.

We wish to show that $C_1 \parallel C_2$ with these event conditions
has the same behaviour as $(C_1 \parallel C_2) \syncstar C_m$. The
labelled transition system of $C_1 \parallel C_2$ has nine derivatives
(including itself).
Eight of these derivatives have the same derivatives in the labelled
transition systems of $(C_1 \parallel C_2) \syncstar C_m$, if we
drop the contribution by derivatives of $C_m$. The derivative that
does not appear is $C'_1 \parallel C'_2$ where each controller has
performed an $\ev{check}_i$ action and can then perform a $\ev{remove}_i$
action.  If we can show that the value of the variable $M$ ensures
that this derivative cannot occur, then it is possible to contruct an
isomorphism between the two LTSs and also to conclude that the two
systems (one using $(C_1 \parallel C_2) \syncstar C_m$ as the
controller, and one using $C_1 \parallel C_2$ with additional
modification to the new variable $M$) are stochastic system bisimilar
with respect to equality. To see that the derivative $C'_1 \parallel
C'_2$ is not reachable from the initial state, consider that if the
first machine has performed $\ev{check}_1$ then $M$ now has value
1, and it is not possible for the second machine to perform
$\ev{check}_1$ because it has the guard that requires $M$ to be
zero. $M$ is only reset to zero when the event $\sev{remove}_1$
ends. Hence it is not possible for $C_2$ to become $C'_2$ until
$C'_1$ has become $C''_1$. This ensures $C'_1 \parallel C'_2$
cannot happen. A similar argument applies if the second machine
executes $\ev{check}_2$ first. Thus the bisimilarity is established.

%\noindent\emph{COMMENT. This can be done in one of two ways, either directly
%or by proving Theorem \ref{thm:equivcon}, that uncontrolled
%systems with the same prefixes and $\sim_s^{\doteq}$-equivalent
%controllers provide controlled systems that are
%$\sim_s^{\doteq}$-equivalent.}

Next, we consider the use of the bisimulation $\sim_s^{\doteq}$
where by using $\doteq$ we require that flows in state have a weaker
form of equivalence than equality. We illustrate this through
allowing different arrival rates for the feeds into the pool.  Let
$Sys_{a_1,a_2,a_3}$ be the system such that \begin{eqnarray*} \Feed_1
& \rmdef & \ev{init}\pc(p_1,a_1,\const).\Feed_1 +
		 \ev{full}\pc(p_1,0,\const).\Feed_1 \\
\Feed_2 & \rmdef & \ev{init}\pc(p_2,a_2,\const).\Feed_2 +
		 \ev{full}\pc(p_2,0,\const).\Feed_2  \\
\Feed_3 & \rmdef & \ev{init}\pc(p_3,a_3,\const).\Feed_3 +
		 \ev{full}\pc(p_3,0,\const).\Feed_3
\end{eqnarray*} Then ${\Sys_{a_1,a_2,a_3} \syncstar \ev{init}.\Con_1}
\sim_s^{\doteq} {\Sys_{b_1,b_2,b_2} \syncstar \ev{init}.\Con_1}$
whenever $a_1 + a_2 + a_3 = b_1 + b_2 + b_3$. The ODE that describes
the change in the amount of items in the pool is $dP/dt =  a_1 +
a_2 + a_3 = k$ and $\ad(\sigma,P,\const) = \ad(\tau,P,\const)$.
Hence as long as the systems being compared have the same value for
$k$, the behaviour will be bisimilar. Replacing the $Feed_i$ by a
single subcomponent \begin{eqnarray*} \Feed & \rmdef &
\ev{init}\pc(p,a,\const).\Feed +
		 \ev{full}\pc(p,0,\const).\Feed
\end{eqnarray*} with $\iv(p)=P$, also provides a bisimilar model
with respect to $\doteq$ as long as $a=k$.

%This illustrates that in certain models, $\doteq$
%can preserve $\Gamma$, although it does not always.

\begin{figure}
\begin{center}
\begin{tabular}{cc}
\hspace*{-0.9cm}
\includegraphics[width=8.5cm]{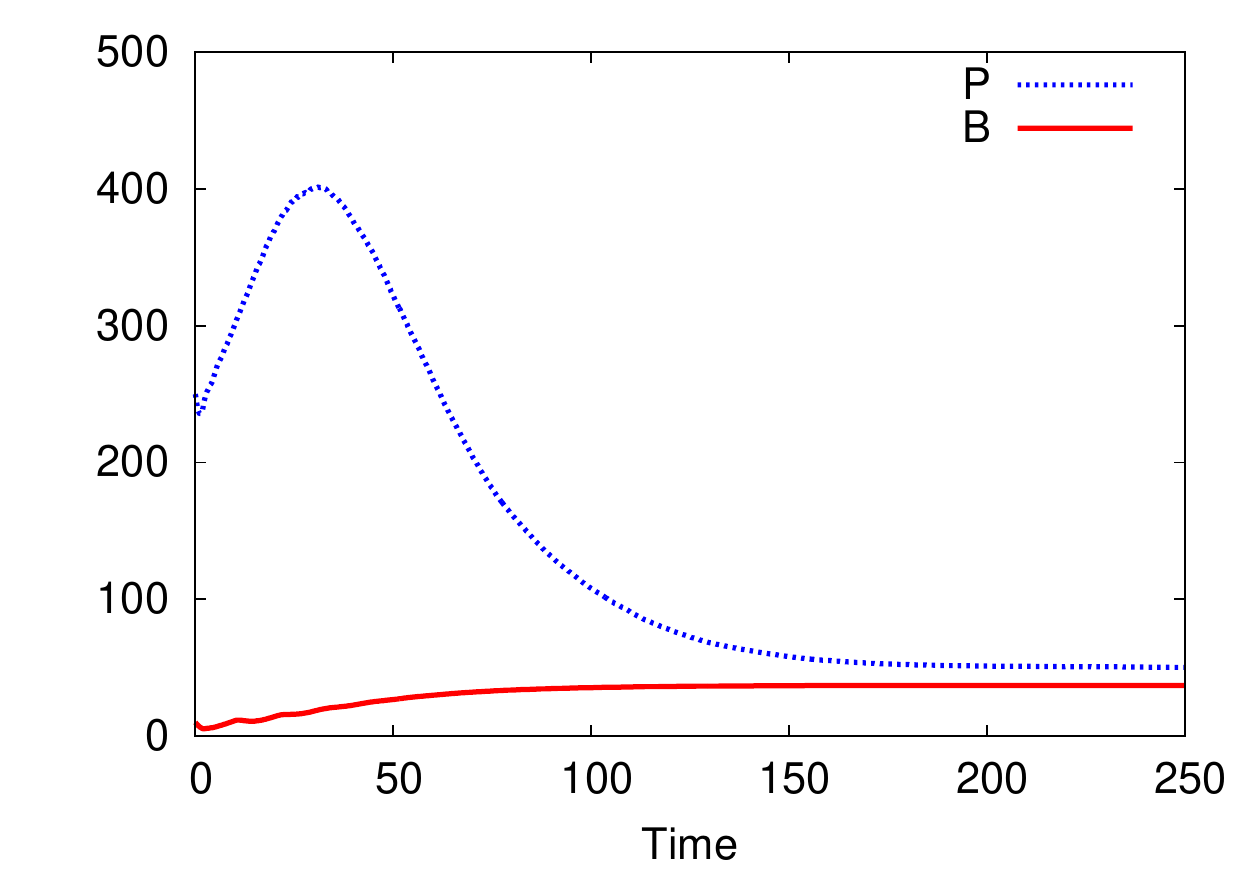}
&
\hspace*{-0.75cm}
\includegraphics[width=8.5cm]{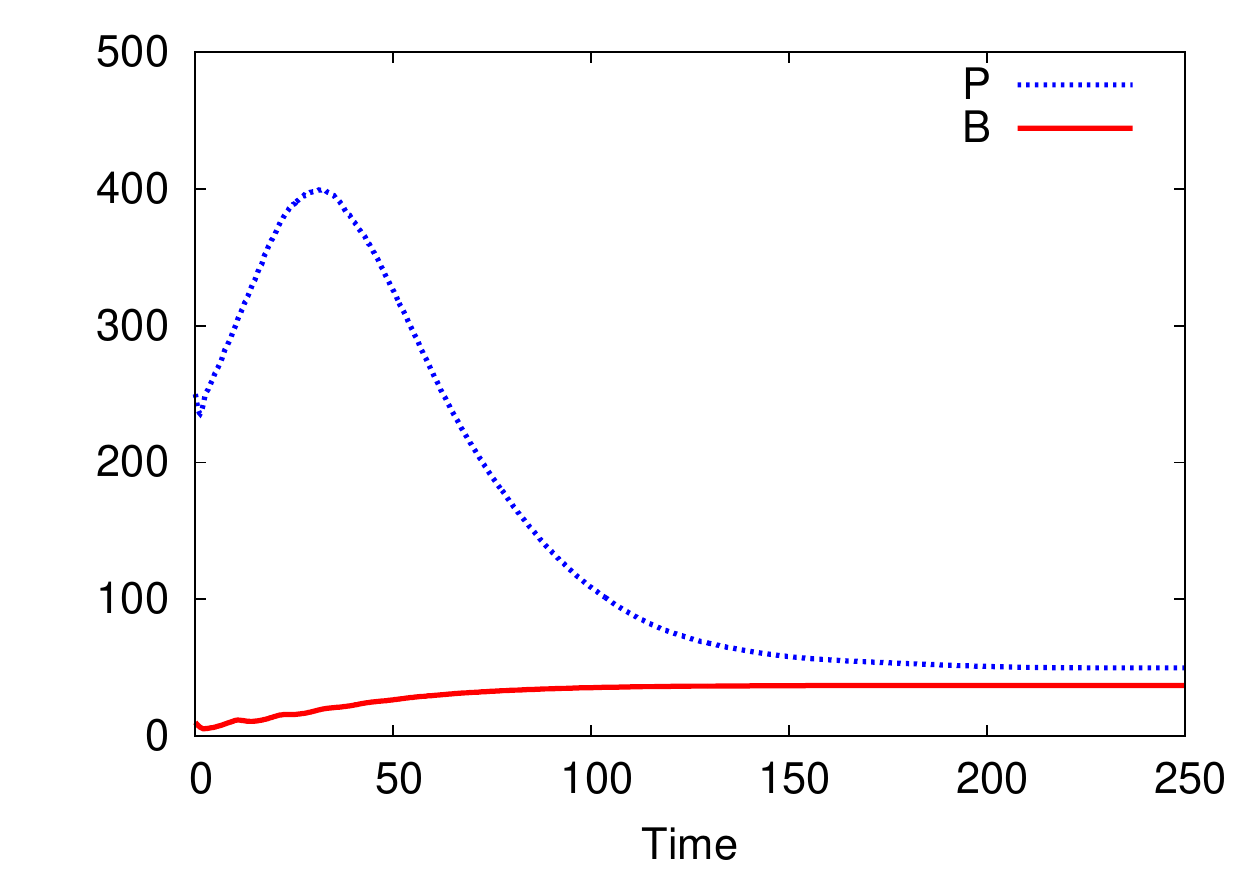}
\end{tabular}
\end{center}
\caption{Average values for $P$ and $B$ for $\Assem$ over 10000
simulations
(left) and average values for $P$ and $B$ for $\Assem_D$ over 10000
simulations
(right) using the same parameters as in Figure~\ref{fig:parrun}.}
\label{fig:SimCon2}
\end{figure}

%\newpage
\section{Related work}

As described previously \cite{HYPE-journal}, HYPE takes a finer
grained, less monolithic approach than the other process algebras
for hybrid systems \cite{BergBM:05a,BeekBMRRS:06a,RounRS:03a,CuijCR:05a}
because it enables the modelling of individual flows. \cite{KhadK:06a}
compares these other process algebras based on the train gate
controller example and for each of these process algebras, complete
ODEs appear in the syntactic description of the system.  Likewise,
hybrid automata \cite{HenzH:96a} require explicit definition of
ODEs, and are less compositional than HYPE since the product of
hybrid automata, requires disjoint variable sets. By contrast, HYPE
product \cite{HYPE-journal} does not require this and hence product
construction allows shared variables in each component and hence
richer behaviour through interaction. Other compositional hybrid
formalisms such as \textsc{Charon} \cite{AlurAGLS:06b}, \textsc{Shift}
\cite{DeshDGV:96a}, and HyCharts \cite{GrosGS:02a} do not map
directly to hybrid automata, compared to HYPE.  Hybrid action systems
\cite{RonkRRS:03a}, Hybrid Interacting Maude \cite{FadlFOA:11a} and
bond graphs \cite{Paynter61,Cuijpers08} are other approaches that
take a less monolithic approach.

Stochastic HYPE retains the fine grained approach to modelling
flows, and the more expressive product construction, as well as
adding the ability to model stochastic events. To the best of our
knowledge, it is the only process algebra with these modelling
capabilities.

Other formalisms for stochastic hybrid modelling include TDSHA
\cite{PA:Bortolussi:2010:HybridDynamicsStochProg:TCS,BorP11}, PDMP
\cite{Davis93}, stochastic hybrid automata (SHA) \cite{BujoBL:04a}
and the recent stochastic extension of UPPAAL \cite{DaviDDLLMPS:12a}.
Stochastic HYPE provides a language for reasoning about some of these
formalisms since these are the semantic objects described by the
stochastic HYPE syntax. TDSHA were developed to provide a
transition-focussed approach to PDMP and hence providing a more
consistent manner of treating the three different types of behaviour
considered: continuous, stochastic and instantaneous behaviour.

Research into making stochastic hybrid systems modelling compositional
has been considered \cite{StruSJS:03a}, %Strubbe \emph{et al.},
where Communicating Piecewise Deterministic Markov Processes (CPDP)
are introduced.  This is an automata-based formalism which models
a system as interacting automata.  Their chosen level of abstraction
is somewhat lower level than ours, comparable with TDSHA. In CPDP,
as in stochastic HYPE, instantaneous transitions may be triggered
either by conditions of the continuous variables (boundary-hit
transitions) or by the expiration of a stochastically determined delay
(Markov transitions).  Interaction between automata is based on
\emph{one-way} synchronisation: in each interaction one partner is
active while the other is passive. In stochastic HYPE, instead, all
components may be regarded as active with respect to each transition
in which they participate, as activation conditions are specified
uniquely in the model. Components participating in a discrete
transition are determined by the construction of the stochastic
HYPE model, where the synchronisation set $L$ in $\!\!{\smsync{L}}\!\!$
specifies which actions must be shared.

The synchronization mechanics of CPDP has been extended by
\cite{StruSS:05b}, introducing an  operator which  exploits all
possible interactions of active and passive actions. \cite{StruSS:05a}
define a notion of bisimulation for both PDMPs and CPDPs and show
that if CPDPs are bisimilar then they give rise to bisimilar PDMPs.
Furthermore the equivalence relation is a congruence with respect
to the composition operator of CPDPs. The definition of bisimilarity
presented in Appendix~\ref{sec:bisimforTDSHA} is based on these
definitions for PDMPs and CPDPs.

Stochastic HYPE has been used to model various systems, including
an orbiter \cite{BortBGH:11a}, a stochastic version of the train
gate \cite{HYPE-journal} and opportunistic networks \cite{BortBGH:12a}.
More recently, a large scale model of ZebraNet \cite{JuanJOWMPR:02a},
a wildlife monitoring project where individual zebra are nodes in
an opportunistic network, has been developed which includes
2-dimensional animal movement, animal behaviour, and different
opportunistic network protocols \cite{FengF:12a}. Stochastic HYPE
has also been used to provide hybrid semantics for the biological
process algebra Bio-PEPA \cite{GalpG:13a} where reactions are treated
stochastically or deterministically depending on species quantities
or reaction rates.

\section{Conclusion}

In this paper we have presented an extension of the hybrid process
algebra HYPE, in which non-urgent events fire at exponentially
distributed random times.  Although syntactically the modification
that this entails is minimal with respect to the original version
of HYPE, the semantics of the language is considerably enriched and
quantified analysis of the modelled behaviour becomes available
purely based on the model.  The stochastic hybrid systems obtained
from stochastic HYPE models fall in the class of Piecewise Deterministic
Markov Processes (PDMP) \cite{Davis93}.  Here we have shown how
such a semantics can be derived through the intermediary of
Transition-Driven Stochastic Hybrid Automata (TDSHA)
\cite{PA:Bortolussi:2010:HybridDynamicsStochProg:TCS}.  The mapping
that we present from stochastic HYPE to TDSHA differs from that
originally presented in \cite{BortBGH:11a}, as we work at the level
of the labelled transition system generated by the operational
semantics.  Nevertheless we show that the obtained TDSHA in each
case is equivalent (in Appendix~\ref{sec:SOSmapping}).  We have discussed a
number of bisumulation equivalences for stochastic HYPE with a
particular focus on notions that are pragmatic and coincide with
intuitive ideas of when behaviours coincide.  Furthermore we have
illustrated these with a case study based on an assembly line.

\paragraph{Acknowledgments:}
This work is partially supported by the EU project QUANTICOL, 600708.

% Bibliography
\bibliographystyle{../eptcs}
\bibliography{../shype-tr}

\begin{appendix}

\newpage

\section{Congruence}
\label{sec:congruence}

\begin{theorem}
\label{congruence}
$\sim_s^\equiv$ is a congruence for Influence, Choice and
Cooperation, if $\equiv$ is preserved by updates and $\Gamma$.
\end{theorem}
\begin{proof}
Let $P_1 \sim_s^\equiv P_2$. 
\begin{description}
\item[Prefix with influence] We have the transition
$\cf{\ssev{a}\pc(\iota,r,I).P_1,\sigma} \lsexrightarrow{\ssev{a}}
\cf{P_1,\sigma[\iota \mapsto (r,I)]}$ and likewise the transition
$\cf{\ssev{a}\pc(\iota,r,I).P_2,\tau} \lsexrightarrow{\ssev{a}}
\cf{P_2,\tau[\iota \mapsto (r,I)]}$. Letting $\sigma \equiv \tau$,
and since $\equiv$ is preserved by updates, we know that $\sigma[\iota
\mapsto (r,I)]\equiv\tau[\iota \mapsto (r,I)]$.  Moreover, for all
$P$, $\sigma$ and $C$, $r(\cf{\sev{a}.P,\sigma},\sev{c},C) =
\act(\sev{a})$ if $\sev{a} = \sev{c}$, otherwise
$r(\cf{\sev{a}.P,\sigma},\sev{c},C) = 0$.  Hence we can conclude
that $\ssev{a}\pc(\iota,r,I).P_1 \sim_s^\equiv \ssev{a}\pc(\iota,r,I).P_2$.

\item[Prefix without influence] We know $\cf{\ev{a}.P_1,\sigma}
\lsexrightarrow{\ev{a}} \cf{P_1,\sigma}$ and $\cf{\ev{a}.P_2,\tau}
\lsexrightarrow{\ev{a}} \cf{P_2,\tau}$ with $\sigma \equiv \tau$.
Furthermore $r(\cf{\sev{a}.P,\sigma},\sev{c},C) = \act(\sev{a})$
if $\sev{a}=\sev{c}$ (otherwise $0$), for any $P$ and $\sigma$.
Since $P_1 \sim_s^\equiv P_2$, then we can conclude that $\ssev{a}.P_1
\sim_s^\equiv \ssev{a}.P_2$.

\item[Choice] First, if $\cf{P_1,\sigma} \lsexrightarrow{\ev{a}}
\cf{P'_1,\sigma'}$, then we have the transition $\cf{P_1+Q,\sigma}
\lsexrightarrow{\ev{a}} \cf{P'_1,\sigma'}$ and since $P_1 \sim_s^\equiv
P_2$, $\cf{P_2,\tau} \lsexrightarrow{\ev{a}} \cf{P'_2,\tau'}$ with
$P'_1 \sim_s^\equiv P'_2$ and $\sigma'\equiv\tau'$, and hence
$\cf{P_2+Q,\tau} \lsexrightarrow{\ev{a}} \cf{P'_2,\tau'}$ as required.
Second, if $\cf{Q,\sigma} \lsexrightarrow{\ev{a}} \cf{Q',\sigma''}$
then also $\cf{Q,\tau} \lsexrightarrow{\ev{a}} \cf{Q',\tau''}$ for
some $\tau''\equiv\sigma''$ and $\cf{P_2+Q,\tau} \lsexrightarrow{\ev{a}}
\cf{Q',\tau''}$. For stochastic events, $r(\cf{P_1+Q,\sigma},\sev{a},C)
= r(\cf{P_1,\sigma},\sev{a},C) + r(\cf{Q,\sigma},\sev{a},C)$. Since
$P_1 \sim_s^\equiv P_2$,
$r(\cf{P_1,\sigma},\sev{a},C)=r(\cf{P_2,\tau},\sev{a},C)$. Moreover
$r(\cf{Q,\sigma},\sev{a},C)=r(\cf{Q,\tau},\sev{a},C)$ and hence, the
conclusion is that
$r(\cf{P_1+Q,\sigma},\sev{a},C) = r(\cf{P_2+Q,\tau},\sev{a},C)$.

\item[Cooperation] We need to show that $B = \{ (P_1 \smsync{L} Q,
P_2 \smsync{L} Q) | P_1 \sim_s^\equiv P_2 \}$ is a system bisimulation
with respect to $\equiv$.  There are three cases to consider. First,
if $\cf{P_1,\sigma} \lsexrightarrow{\ev{a}} \cf{P'_1,\sigma'}$ with
$\ev{a} \not\in L$ then $\cf{P_1 \smsync{L} Q, \sigma}
\lsexrightarrow{\ev{a}} \cf{P'_1 \smsync{L} Q, \sigma'}$.  Since
$P_1 \sim_s^\equiv P_2$, for $\tau$ such that $\sigma\equiv\tau$,
$\cf{P_2,\tau} \lsexrightarrow{\ev{a}} \cf{P'_2,\tau'}$ with
$\sigma'\equiv\tau'$ and $P'_1 \sim_s^\equiv P'_2$, and hence
$\cf{P_2 \smsync{L} Q,\tau} \lsexrightarrow{\ev{a}} \cf{P'_2
\smsync{L} Q,\tau'}$.  Next, considering stochastic events, we need
to show that $r(\cf{P_1 \smsync{L} Q,\sigma},\sev{a},C) = r(\cf{P_2
\smsync{L} Q,\tau},\sev{a},C)$ for all equivalence classes $C \in
(\CF/B)/\equiv$.  Since $r(\cf{P_1 \smsync{L} Q,\sigma},\sev{a},C)
= r(\cf{P_1,\sigma},\sev{a},C)+r(\cf{Q,\sigma},\sev{a},C)$, we have
the result. The second case where $\ev{a} \not\in L$ and $\cf{Q,\sigma}
\lsexrightarrow{\ev{a}} \cf{Q',\sigma''}$ is proved in a similar
fashion.

Third, consider, $\ev{a} \in L$ and $\cf{P_1,\sigma}
\lsexrightarrow{\ev{a}} \cf{P'_1,\sigma'}$ and $\cf{Q,\sigma}
\lsexrightarrow{\ev{a}} \cf{Q',\sigma''}$, then we have the transition
$\cf{P_1
\smsync{L} Q, \sigma} \lsexrightarrow{a} \cf{P'_1 \smsync{L} Q',
\Gamma(\sigma,\sigma',\sigma'')}$. Since $P_1 \sim_s^\equiv P_2$,
and letting $\sigma\equiv\tau$, then $\cf{P_2,\tau}
\;\lsexrightarrow{\ev{a}}\; \cf{P'_2,\tau'}$ with $P'_1 \sim_s^\equiv
P'_2$ and $\sigma'\equiv\tau'$.  Also $\cf{Q,\tau} \lsexrightarrow{\ev{a}}
\cf{Q',\tau''}$ with $\sigma''=\tau''$.  Hence we have the transition
$\cf{P_2 \smsync{L} Q,\tau} \:\lsexrightarrow{\ev{a}} \cf{P'_2
\smsync{L} Q',\Gamma(\tau,\tau',\tau'')}$ with
$\Gamma(\sigma,\sigma',\sigma'') \equiv \Gamma(\tau,\tau',\tau'')$
as $\equiv$ is preserved by $\Gamma$.  For stochastic transitions
$\sev{a} \in L$,  we know that $r(\cf{P_1 \smsync{L} Q,\sigma},\sev{a},C)
= \act(\sev{a}) \cdot \sum
\lms\multi(\cf{P_1,\sigma}\lsexrightarrow{\sev{a}} F) \mid F \in C
\rms \cdot \sum \lms\multi(\cf{Q,\sigma}\lsexrightarrow{\sev{a}}
F) \mid F \in C \rms$, ensuring preservation of multiplicities.
Since $P_1 \sim_s^\equiv P_2$,
$r(\cf{P_1,\sigma},\sev{a},C)=r(\cf{P_2,\tau},\sev{a},C) =
\act(\sev{a})\cdot\sum \lms\multi(\cf{P_2,\tau}\lsexrightarrow{\sev{a}}
F) \mid F \in C \rms$, and hence we know $r(\cf{P_2 \smsync{L}
Q,\tau},\sev{a},C) = r(\cf{P_2 \smsync{L} Q,\sigma},\sev{a},C)$.
\end{description}
\end{proof}

\section{Equivalence of the two semantics}
\label{sec:equiv}

\begin{theorem}
Let 
$(\ConSys,\,\V,\,\IN,\,\IT,\,\Ev_c,\,\Ev_s,\,\Ac,\,\ec,\,\iv,\,\EC,\,\ID)$
be a stochastic HYPE model
then the TDSHA obtained via the operational semantics is the same as the
TDSHA obtained by the compositional mapping \cite{BortBGH:11a} 
when only reachable
modes are considered.
\end{theorem}
\begin{proof}
Let $\cf{P',\sigma_1} \in \ds(P)$ and $\cf{Q',\sigma_2} \in \ds(Q)$
with $P' \thicksim_\mathbf{sm}^{\doteq} Q'$ then we can conclude
that $\sigma_1 \doteq \sigma_2$. We need to identify the ODEs
associated with each configuration. Since both of these are modes
in the respective TDSHAs, we can identify their respective multisets
of continuous transitions.

If $\sigma_1 = \{\iota_i\mapsto (r_i,I_i)~|~i=1,\ldots k\}$ then
the multiset of continuous transitions associated with $\cf{P',\sigma_1}$
is $\lms (\blang P, \sigma_1
\brang,\vr{1}_{\iv(\iota_i)},r_i.\assign{I_i}) \mid i=1,\ldots k \rms$.
Hence, for a given $V_j$, we obtain the ODE
\[
\frac{dV_j}{dt} = \sum_{i=1}^k \lms r_i\assign{I_i(\V)} \: \big| \:
   \iv(\iota_i) = V_j \rms \]
This can be written more generally as follows.
%Similary, if $\sigma_2 = \{\gamma_i\mapsto (S_i,J_i)~|~i=1,\ldots
%j\}$ then the multiset of continuous transitions associated with
%$\cf{Q,\sigma_2}$ is $ \lms (\blang Q, \sigma_2
%\brang,\vr{1}_{\iv(\gamma_i)},s_i.\assign{J_i}) \mid i=1,\ldots j
%\rms$
\[ P_{\sigma_1}=
  \Bigl\{ \frac{dV}{dt} = \sum \lms r\assign{I(\V)}
                        \: \big| \:
                        \iv(\iota) = V,
                        \sigma_1(\iota) = (r,I(\V))
\rms \Bigm| V \in \V \Bigr\} \]
Similarly, for $\cf{Q',\sigma_2}$ we have the following ODEs.
\[ Q_{\sigma_2}=
  \Bigl\{ \frac{dV}{dt} = \sum \lms r\assign{I(\V)}
                        \: \big| \:
                        \iv(\iota) = V,
                        \sigma_2(\iota) = (r,I(\V))
\rms \Bigm| V \in \V \Bigr\} \]
Since $\sigma_1 \doteq \sigma_2$,
for each $\iota$, whenever $\sigma_1(\iota) = (r_1,I_1(\V))$ and
$\sigma_2(\iota) = (r_2,I_1(\V))$ then $\assign{I_1(\V))}
=
\assign{I_2(\V)}$
and for each $V \in \V$, $I \in \IN$,
$\ad(\sigma_1,V,I(\V))=\ad(\sigma_2,V,I(\V))$.
Consider
\begin{eqnarray*}
\frac{dV}{dt} & = & \sum \lms r\assign{I(\V)} \: \big| \:
   \iv(\iota) = V, \sigma_1(\iota) = (r,I(\V)) \rms \\
         & = & \sum_{I(\V)} \assign{I(\V)}.\sum
\lms \:\big|\:
   \iv(\iota) = V, \sigma_1(\iota) = (r,I(\V)) \rms \\
         & = & \sum_{I(\V)}
\assign{I(\V)}.\ad(\sigma_1,V,I(\V)) 
        \: = \: \sum_{I(\V)}
\assign{I(\V)}.\ad(\sigma_2,V,I(\V)) \\
         & = & \sum_{I(\V)} \assign{I(\V)} \sum
\lms r\:\big|\:
   \iv(\iota) = V, \sigma_2(\iota) = (r,I(\V)) \rms \\
         & = & \sum \lms r\assign{I(\V)} \: \big| \:
   \iv(\iota) = V, \sigma_2(\iota) = (r,I(\V)) \rms 
\end{eqnarray*}
since for each $I(\V)$, the rates obtained from the states sum to
the same value. Therefore $P_{\sigma_1} = Q_{\sigma_2}$. 
\end{proof}

\section{Comparison of TDSHA mappings}
\label{sec:SOSmapping}

To show that the mapping which defined the semantics for stochastic
HYPE \cite{BortBGH:11a} is the same as the mapping from SOS semantics
presented here, the definition of the product of two TDHSAs is
required as well as the mapping definition. We will refer to the
mapping in \cite{BortBGH:11a} as the \emph{compositional mapping}
denoted $\TDSHA{}$, and will use \emph{SOS mapping}, denoted
$\TDSHA{\mathit{SOS}}$, for the mapping in this paper. The compositional
mapping did not use random resets but we do so here.

For the product, we require a consistency definition relating to
resets, to ensure that no resets clash by attempting to set the
same variable to two different values.  Hence, we say that two
transitions $\delta_1,\delta_2$ (either both discrete or both
stochastic) are \emph{reset-compatible} if and only if
$\e{\delta_1}\neq\e{\delta_2}$ or $\res{\delta_1}\wedge\res{\delta_2}\neq
\mathit{false}$. For random resets, this requires that any random
variable (that necessarily appears in both resets) is drawn from the same
distribution. Two TDSHA are reset-compatible if and only if all
their discrete or stochastic transitions are pairwise reset-compatible.
A similar notion is required for the initial conditions: Two TDSHA
are \emph{init-compatible} if and only if, given initial conditions
$\initial{1} = (\initmode{1},\inp{1})$ and $\initial{2} =
(\initmode{2},\inp{2})$, then $\inp{1}\wedge\inp{2}\neq \mathit{false}$.

In the definition of TDSHA product which follows, the set of
continuous transitions in a mode $\mode{}=(\mode{1},\mode{2})$
contains all continuous transitions of $\mode{1}$ and all those of
$\mode{2}$. The set of instantaneous transitions $\TD{}$ is the
union of non-synchronized instantaneous transitions $\TD{\mathit{NS}}$
and of synchronized ones $\TD{S}$, and during synchronization, a
conservative policy is applied by taking the conjunction of guards
and resets, and by taking the minimum of weights. Similarly, the
set of stochastic transitions is defined as $\TS{} =
\TS{\mathit{NS}}\cup\TS{S}$.  In the synchronization of stochastic
transitions, we use the fact that the rate is the same for all
transitions labelled by the same event, as required by the consistency
condition.

\begin{definition}\label{def:TDSHAproduct}
Let $\TDSHA{i}=\TDSHAtuple{i}$, $i=1,2$ be
two reset-compatible and init-compatible TDSHA, and let $L\subseteq
\events{1}\cap\events{2}$ be the synchronization set. The
$L$-product $\TDSHA{}  = \TDSHA{1}\tprod{L} \TDSHA{2}=
(Q,\vr{X},\TC,\TD,$ $\TS,\mathrm{init},\mathcal{E})$
is defined by
\begin{enumerate}
  \item $\modes{} = \modes{1}\times\modes{2}$;
  \item $\vrX{} = \vrX{1}\cup\vrX{2}$;
  \item $\events{} = \events{1}\cup\events{2}$;
  \item $\initial{} = (\initmode{},R_{\inp{}})$, where $\initmode{} =
  (\initmode{1},\initmode{2})$ and $R_{\inp{}} =
  R_{\inp{1}}\wedge R_{\inp{2}}$.
  \item $\TC{}
  = \lms\left((\mode{1},\mode{2}),\stoich{},\rf{}
\right)~|~\mode{1}\in\modes{1},\mode{2}\in\modes{2},(\mode{1},\stoich{},\rf{})\in\TC{1}\vee
(\mode{2},\stoich{},\rf{})\in\TC{2} \rms$.
  \item $\TD{} = \TD{NS} \cup \TD{S}$ where

%\hspace*{-0.2cm}
$\begin{array}{rcl}  
  \TD{NS} & = &
  \{\left((\mode{1},\mode{2}),(\mode{1}',\mode{2}'),\guard{},\res{},\priority{},\e{}\right)~|~\\
   & & \hspace*{2.6cm}
(\mode{i},\mode{i}',\guard{},\res{},\priority{},\e{})\in\TD{i}\wedge\mode{j}=\mode{j}'\in\modes{j}\wedge
i\neq j \wedge\e{}\not\in
  S\}, \\
  \TD{S} &=&
\{\left((\mode{1},\mode{2}),(\mode{1}',\mode{2}'),\guard{1}\wedge\guard{2},\res{1}\wedge\res{2},\min\{\priority{1},\priority{2}\},\e{}\right)~|~
\\
   & & \multicolumn{1}{r}{
  (\mode{1},\mode{1}',\guard{1},\res{1},\priority{1},\e{})\!\in\!\TD{1}\wedge
  (\mode{2},\mode{2}',\guard{2},\res{2},\priority{2},\e{})\!\in\!\TD{2}
   \wedge\e{}\!\in\!S\}.}
\end{array}$

  \item $\TS{} = \TS{NS} \cup \TS{S}$ where

  $\begin{array}{rcl}
\TS{NS} & = &
  \lms\left((\mode{1},\mode{2}),(\mode{1}',\mode{2}'),\guard{},\res{},\rf{},\e{}\right)~|~\\
   & &\hspace{2.8cm}
(\mode{i},\mode{i}',\guard{},\res{},\rf{},\e{})\in\TS{i}\wedge\mode{j}=\mode{j}'\in\modes{j}\wedge
i\neq j\wedge\e{}\not\in
  S\rms, \\
\TS{S} & = &
  \lms\left((\mode{1},\mode{2}),(\mode{1}',\mode{2}'),\guard{1}\wedge\guard{2},\res{1}\wedge\res{2},\rf{},\e{}\right)~|~\\
  & & \multicolumn{1}{r}{
(\mode{1},\mode{1}',\guard{1},\res{1},\rf{},\e{})\in\TS{1}\wedge
  (\mode{2},\mode{2}',\guard{2},\res{2},\rf{},\e{})\in\TS{2}\wedge\e{}\in
  S\rms.}
  \end{array}$
\end{enumerate}
\end{definition}

\subsection{Compositional mapping from HYPE to TDSHA}

The mapping $\TDSHA{}$ that is now defined, works compositionally,
by associating a TDSHA with each single subcomponent and with each
piece of the controller, then taking their synchronized product
according to the synchronization sets of the stochastic HYPE system.  Guards,
rates, and resets of discrete edges will be incorporated in the
TDSHA of the controller, while continuous transitions will be
extracted from the uncontrolled system.  Consider a stochastic HYPE model
$(\ConSys,\V,\IN,\IT,\Ev_d,\Ev_s,\Ac,\ec,\iv,\EC\!,\ID)$ with
$\ConSys ::= \Sigma \smash{\sync{L}} \ev{init}.\Con$.

To define the TDSHA of the uncontrolled system $\Sigma$, consider
a subcomponent $S = \sum_{i=1}^k \ssev{a}_i\pc\alpha_i.S +
\ev{init}\pc\alpha.S$.  Each element of $\in(S)$ generates a mode
in the TDSHA of $S$.  Moreover, in each such mode, the only continuous
transition will be the one that is derived from this influence. As
for discrete edges, since response to all events is always enabled,
there will be an outgoing transition for each event appearing in
$S$ in each mode of the associated TDSHA. The target state of the
transition will be the mode corresponding to the influence following
the event.  Resets and guards will be set to $true$, as event
conditions are associated with the controller. Rates of transitions
derived from stochastic events $\sev{a}\in \Ev_s$  will be set to
$\act(\sev{a})$, as required by the consistency condition of TDSHA.
Finally, weights will be set to 1, while the initial mode will be
deduced from the $\ev{init}$ event. Formally, this is defined as
follows.

\begin{definition}\label{def:TDSHAsubcompoenent}
Let $S(\V)\!=\!\sum_{i=1}^k \ssev{a}_i\pc\alpha_i.S(\V)\!+\!
\ev{init}\pc\alpha.S(\V)$ be a subcomponent of the stochastic HYPE model
$(\ConSys,\V,\IN,\IT,$ $\Ev_d,\Ev_s,\Ac,$ $\ec,\iv,\EC,\ID)$. 
Then $\TDSHA{}(S) = \TDSHAtuple{}$, the TDHSA associated with $S$ is defined
by
\begin{enumerate}
  \item $\modes{} = \{\mode{\alpha}~|~\alpha\in\is(S)\}$; $\vrX{} =
  \V$; $\events{} = \Ev_d\cup \Ev_s$;
  \item $\initial{} = (\mode{\alpha},true)$, where $S =
  \ev{init}\pc\alpha.S + S'$;
  \item $\TC{} = \lms(\mode{\alpha},\vr{1}_{\iv(\iname{i}_S)},r\cdot\eval{I})~|~\alpha=(\iname{i}_S,r,I)\in\is(S)
  \rms$, 
  \item $\TD{} =
\{(\mode{\alpha_1},\mode{\alpha_2},1,true,true,\ev{a})~|~\ev{a}\in\es(S)\cap\Ev_d\wedge\alpha_1\in\is(S)\wedge
S=\ev{a}\pc\alpha_2.S + S'\}$
  \item $\TS{} =
\lms(\mode{\alpha_1},\mode{\alpha_2},true,true,\act(\sev{a}),\sev{a})~|~\sev{a}\in\es(S)\cap\Ev_s\wedge\alpha_1\in\is(S)\wedge
S=\ev{a}\pc\alpha_2.S + S'\rms$
\end{enumerate}
\end{definition}

Once a TDSHA is generated for each subcomponents, the TDSHA of the
full uncontrolled system can be built by applying the product
construction of TDSHA.

\begin{definition}\label{def:TDSHAuncontrolled}
If $\Sigma\!\rmdef\!S_1(\V) \syncstar \ldots \syncstar
S_s(\V)$
then $\TDSHA{}(\Sigma) \!=\!
\TDSHA{}(S_1(\V))\tprod{*} \ldots \tprod{*}\TDSHA{}(S_s(\V))$.

\end{definition}
Dealing with the controller is simpler, as controllers are finite
state automata which impose causality on the happening of events.
Event conditions are assigned to edges of TDSHA associated with
controllers. All events will be properly dealt with through this
construction, as they all appear in the controller since the stochastic HYPE
model is well-defined.

First, consider a sequential controller $M = \sum_{i} \ssev{a}_i.M_i$.
The derivative set of $M$ is defined recursively by $ds(M) = \{M\}\cup
\bigcup_i ds(M_i)$.
%, where two summations coincide if they are equal
%.up to permutation of addends. 
%\commentv{Check this -- need multiset definition}

\begin{definition}\label{def:TDSHAseqContr} Let
$(\ConSys,\V,\IN,\IT,\Ev_d,\Ev_s,\Ac,\ec,\iv,\EC,\ID)$ be a stochastic HYPE
model with sequential controller $M$.
Then $\TDSHA{}(M) = \TDSHAtuple{}$, 
the TDSHA associated with $M$, is defined by
\begin{enumerate}
  \item $\modes{} = \{\mode{M'}~|~M'\in ds(M)\}$; $\vrX{} = \V$; $\events{} = \Ev_d\cup \Ev_s$;
  \item $\initial{} = (\mode{M},R_{\ev{init}})$, where
  $R_{\ev{init}}$ is the reset associated with the $\ev{init}$
  event.
  \item $\TC{} = \emptyset$;
  \item $\TD{} = \{(\mode{M_1},\mode{M_2},1,\act(\ev{a}),R_{\ev{a}},\ev{a})
  ~|~M_1 = \ev{a}.M_2,\ M_1,M_2\in ds(M),\ \ev{a}\in\Ev_d,\
  \ec(\ev{a}) = (\act(\ev{a}),$ $\rs(\ev{a})) \}$;
  \item  $\TS{} = \lms(\mode{M_1},\mode{M_2},true,R_{\sev{a}},\act(\sev{a}),\sev{a})
  ~|~M_1 = \sev{a}.M_2,\ M_1,M_2\in ds(M),\ \sev{a}\in\Ev_s,\
  \ec(\sev{a}) = (\act(\sev{a}),$ $\rs(\sev{a})) \rms$, where
  $\act(\sev{a}):\bbR^{|\V|}\rightarrow\bbR^+$ is the rate of
  the transition;
\end{enumerate}

\end{definition}

\begin{definition}\label{def:TDSHAcontroller}
Let $Con = Con_1\sync{L}\Con_2$ be a controller. The TDSHA of $Con$
is defined recursively as $\TDSHA{}(Con) =
\TDSHA{}(Con_1)\tprod{L}\TDSHA{}(Con_2)$.
\end{definition}

The product construction of Definitions~\ref{def:TDSHAuncontrolled}
and~\ref{def:TDSHAcontroller} are defined because the factor TDSHA
are reset-compatible and init-compatible. This is trivial both for
the uncontrolled system (all resets are $true$) and for the controller
(resets for the same event are equal). Furthermore, stochastic
transitions have consistent rates, as their rate depends only on
the labelling event.

Once the TDSHA of the controller and the uncontrolled system are
constructed, we simply have to take their product.

\begin{definition}\label{def:TDSHAofHYPE}
Let $(\ConSys,\V,\IN,\IT,\Ev_c,\Ev_s,\Ac,\ec,\iv,\EC,\ID)$ be a
stochastic HYPE model, with controlled system
$\ConSys = \Sigma\sync{L} \ev{init}.\Con$. The
TDSHA associated with $\ConSys$ is
$$\TDSHA{}(\ConSys) = \TDSHA{}(\Sigma)\tprod{L}\TDSHA{}(\Con).$$
\end{definition}

The construction generates TDSHAs  with many \emph{unreachable
states} \cite{BortBGH:11a}. This is a consequence of the fact that
sequentiality and causality on actions is imposed just on the final
step, when the controller is synchronized with the uncontrolled
system. Once the TDSHA is constructed, however, it can be pruned
by removing unreachable states.  In order to limit combinatorial
explosion, one can prune TDSHA's at each intermediate stage. A
formal definition of this policy, however, would have made the
mapping from stochastic HYPE to TDSHA much more complex.

\begin{theorem}
Let 
$(\ConSys,\,\V,\,\IN,\,\IT,\,\Ev_c,\,\Ev_s,\,\Ac,\,\ec,\,\iv,\,\EC,\,\ID)$ 
be a stochastic HYPE model then \linebreak
$\TDSHA{}(\ConSys) = \TDSHA{\mathit{SOS}}(\ConSys)$ when only reachable
states are considered.

\end{theorem}
\begin{proof}
In order to prove the theorem, we will exhibit a graph isomophism
between the discrete graph of the two TDSHA $\TDSHA{}(\ConSys)$ and
$\TDSHA{\mathit{SOS}}(\ConSys)$. The additional information labelling
edges of the graph (i.e. guards, resets, rates and priorities) are
automatically the same due to the fact they are defined globally.
Priorities, in particular, are always equal to 1 (and the $\min$
in the product of TDSHA preserves this).

First of all, consider $\TDSHA{}(\ConSys)$ and observe that by
Definitions 22, 23, and 26, each mode of the automaton contains the
set of active influences, which are in bijection with the set of
continuous transitions of $\TDSHA{}(\ConSys)$. Furthermore, by
Definitions 24, 25, and 26, it also contains the current state of
each sequential component of the controller. Hence, we can indicate
a mode $q$ by the tuple $( (\iota_1,r_{1,j_1},I_{1,j_1}),\ldots,(\iota_k,
r_{k,j_k},I_{k,j_k}),Con_{1,i_1},\ldots,Con_{h,i_h})$. On the other
hand, a mode of $\TDSHA{\mathit{SOS}}(\ConSys)$ is of the form
$\cf{\Sigma\syncstar Con,\sigma}$, with $Con =
Con_{1,i_1}\syncstar\ldots\syncstar Con_{h,i_h}$ and where the state
is $\sigma = \{ (\iota_1,r_{1,j_1},I_{1,j_1}),$ $\ldots,(\iota_k,
r_{k,j_k},I_{k,j_k}) \}$.  Therefore, we can define  the function
$\rho$ mapping each $\cf{\Sigma\syncstar Con,\sigma}$ to the tuple $(
(\iota_1,r_{1,j_1},I_{1,j_1}),\ldots,(\iota_k,
r_{k,j_k},I_{k,j_k}),Con_{1,i_1},\ldots,Con_{h,i_h})$.
This function is well-defined, as the derivative of the controller
plus the state uniquely identify the configuration, and it is easily
seen to be a bijection.

In order to show that $\rho$  is a graph isomorphism, we need to
show that if  $\cf{P_1,\sigma_1} \:\smash{\lts{\ssev{a}}}\:
\cf{P_2,\sigma_2}$ is a transition of the LTS, then there is a
discrete or stochastic transition (depending on $\ssev{a}$) in
$\TDSHA{}(\ConSys)$ of the form $(\rho(\cf{P_1,\sigma_1}
),\rho(\cf{P_2,\sigma_2}),\cdot,\cdot,\cdot,\ssev{a})$, with matching
multiplicity.

This can be easily seen by structural induction on agents, focussing
on synchronisation and taking subcomponents or sequential controllers
as base cases.  Inspecting Definitions 22 and 24, it is easy to see
that the previous property holds for the base cases, as those
definitions construct the TDSHA by considering all transitions and
all states of the LTS.

As for synchronization, we just need to notice that the product
construction for TDSHA perfectly matches the SOS rules in terms of
updated components and updated states, hence a simple structural
induction argument will do. Multiplicity of stochastic transitions
is also preserved as the product construction acts on multi-graphs
for what concerns stochastic events.  Hence, we have proved that
$\rho$ is a graph isomorphism. The fact that reachable states are the
same in both graphs then follows by a simple induction on the
distance in such graphs from the initial mode.  
\end{proof}

\section{Bisimulation for TDSHA}
\label{sec:bisimforTDSHA}

We have a number of choices when we consider bisimulations over
TDSHA in light of definitions for PDMP \cite{StruSS:05a}. We can
consider a single discrete jump as a transition and match or ignore
labels, or we can consider a sequence of discrete jumps, ignoring
labels or considering each jump to be labelled with $\ev{\text{$\tau$}}$
except for one labelled with $\ev{a} \neq \ev{\text{$\tau$}}$.  The
first two options can be viewed as strong forms of bisimulation and
the last two as weak forms of bisimulation. One reason for ignoring
labels is then we can compare any equivalence we define with those
for PDMPs which have no labels. However, the focus of this paper
is stochastic HYPE where labels are used, and hence we focus on
those with labels.

We start by giving some definitions that we require for strong
bisimulation. We need a notion of measurable relation which will
allow the definition of function $\match_B$ that provides a bijection
between two quotient spaces and can be used in the definition of
bisimulation.

\subsection{Measurable relations}\label{sec:measurableRelations}
\label{measrel}

In order to define bisimulations for TDSHA, we need to introduce a
notion of measurable relation, taken from~\cite{StruSS:05a}. In the
following, we let $X,Y$ be separable metric spaces and $\calX,\calY$
be the corresponding Borel sigma-algebras, so that $(X,\calX)$ and
$(Y,\calY)$ are two Borel measurable spaces.

We consider a relation $B\subseteq X\times Y$, and we assume that
the projections on the two components coincide with the whole spaces,
i.e. $\{x\in X~|~\exists y\in Y,(x,y)\in B\} = X$ and $\{y\in
Y~|~\exists x\in X,(x,y)\in B\} = Y$. Then we define two equivalence
relations, one on $X$ and one on $Y$. The equivalence relation $B_X$
on $X$ is defined as the transitive closure of the relation
$\{(x_1,x_2)~|~\exists y\in Y,(x_1,y),(x_2,y)\in B\}$.  Similarly
for $B_Y$.

A straightforward property of the equivalence relations induced by
$B$ would be so that the two quotient spaces $X/B = X/B_X$ and $Y/B = Y/B_Y$
are in bijection. In fact, the map $\match_B: X/B\rightarrow Y/B$
such that $\match_B([x]) = [y]$ if and only if $(x,y)\in B$, is a well-defined
bijection.

In the following, we denote with $\pi_X$ the canonical projection
of $X$ onto $X/B$, defined by $\pi_X(x) = [x]$. Similarly for $Y$.

Another property of the equivalence relations induced by $B$ is
that $X/B$ and $Y/B$ inherit the sigma-algebra structure from $X$
and $Y$. In fact, it is straightforward to check that the collection
$\calX/B$ of subsets of  $X/B$, containing the sets $\{[x]\in
A~|~A\in\calX\}$, is a sigma-algebra.

\begin{definition}\label{def:measurableRelation}
The relation $B\subseteq X\times Y$ is \emph{measurable} if and only
if, for each $A\in\calX/B$, it holds that $\match_B(A)\in\calY/B$,
and vice versa.
\end{definition}

We further need the notion of equivalent probability measures on
$X$ and $Y$, with respect to a measurable relation $B$. Essentially,
two probability measures will be equivalent if and only if they will induce
the same probability distribution on the quotient sets $X/B$ and
$Y/B$.

\begin{definition}\label{def:equivalentPD}
Let $P_X$ be a probability measure on $(X,\calX)$ and $P_Y$ be a
probability measure on $(Y,\calY)$. $P_X$ and $P_Y$ are equivalent
with respect to the measurable relation $B$ if and only if, for each
$A\in\calX/B$, it holds that $$P_X(\pi_X^{-1}(A)) =
P_Y(\pi_Y^{-1}(\match_B(A))).$$
\end{definition}

\subsection{Bisimulation}

To define TDSHA bisimulation, definitions are required to probailities
of actions. The two definitions below define transitions that involve
actions from a set $A$, and calculate probabilities for these
transitions.  The set $A$ is used to determine which labels are
matched and for matching of single labels, $\{a\}$ can be used.  We
can also define non-singleton subsets of $A$ but we defer these to
further work.

\begin{definition}
Given a TDSHA $\TDSHA{}=\TDSHAtuple{}$, let
$$\TS{}((\mode{},\vr{x}),A) =
\lms\eta\in\TS~|~\exit{\eta}=\mode{},\e{\eta}\in
A,\guard{\eta}(\vr{x})=true\rms$$ be 
the \emph{set of stochastic
transitions with labels in $A \subseteq \events{}$ active in 
$(\mode{},\vr{x})$}.\\
Furthermore, let 
$$\TD{}((\mode{},\vr{x}),A) =
\lms\delta\in\TD~|~\exit{\delta}=\mode{},\e{\delta}\in A,
\guard{\delta}(\vr{x})=true\rms$$ 
be the \emph{set of instantaneous
transitions with labels in $A \subseteq \events{}$ active in 
$(\mode{},\vr{x})$}.
\end{definition}

Let $\lambda(q,\vr{x}) = \sum \lms \rf{\eta}(\vr{x}) \mid \eta \in
\TS{}((\mode{},\vr{x}),\events{}) \rms$ and let
$\priority{}(q,\vr{x}) = \sum \lms \priority{\delta}(\vr{x}) \mid
\delta\in\TD{}((q,\vr{x}),\events{}) \rms$.

\begin{definition}
The (1-step) \emph{probability of a stochastic transition with a
label in $A$ from $(\mode{},\vr{x})$ to a set $C$} is defined as
\[P_{1s}^{\TS{}}((\mode{},\vr{x}),A,C) = \sum 
\lms
\mathit{Pr}\{(\enter{\eta},\res{\eta}(\vr{x},\mathbf{W})) \in C\} \cdot
\rf{\eta}(\vr{x})/\lambda(\mode{},\vr{x}) 
\mid \eta \in
\TS{}((\mode{},\vr{x}),A) %, (\enter{\eta},\res{\eta}(\vr{x})) \in C
\rms\] 
for $\lambda(\mode{},\vr{x})\neq 0$ and 0 otherwise.

Similarly, the (1-step) \emph{probability of an instantaneous
transition a label in $A$ from $(\mode{},\vr{x})$ to a set $C$} is
defined as 
\[P_{1s}^{\TD{}}((\mode{},\vr{x}),A,C) = \sum \lms
\mathit{Pr}\{(\enter{\delta},\res{\eta}(\vr{x},\mathbf{W})) \in C\} \cdot
\priority{\delta}/\omega(\mode{},\vr{x}) \mid \delta \in
\TD{}((\mode{},\vr{x}),A) %, (\enter{\delta},\res{\delta}(\vr{x})) \in C 
\rms\] for $\omega(\mode{},\vr{x})\neq 0$, and 0 otherwise.
\end{definition}

To improve the readability of the following definitions, we define
the predicate  $$G(\mode{},\vr{x}) = \bigvee \{ g_\delta(\vr{x})
\mid \delta \in \TD{}, \exit{\delta}=\mode{}\},$$ which is true
when at least one guard of an instantaneous transition is true.
Let $\phi(t,\vr{X})$ denote the solution of the ODEs of the TDHSA
taking into account the initial values of variables.

Next, we define a bisimulation that matches on individual labels.

\begin{definition}\label{def:TDSHAlabeledBisimulation}
Let $\TDSHA{i}=\TDSHAtuple{i}$, $i=1,2$ be two TDSHAs, and let the
relation $B \subseteq (\modes{1} \times \mathbb{R}^{n_1}) \times
(\modes{2} \times \mathbb{R}^{n_2})$ be a measurable relation. $B$
is a \emph{labelled TDHSA bisimulation} for $\TDSHA{1}$ and
$\TDSHA{2}$ whenever for all
$((\mode{1},\vr{x}),(\mode{2},\vr{x})) \in B$,
\begin{enumerate}
\item Assuming two smooth output functions $\out_i: \mathbb{R}^{n_i}
\rightarrow \mathbb{R}^{m}$, for $i=1,2$ with $m \leq n_i$ then
$\out_1(\vr{x_1}) = \out_2(\vr{x_2})$
\smallskip
\item $\lambda(\mode{1},\vr{x_1}) = \lambda(\mode{2},\vr{x_2})$
\smallskip
\item $G_1(\mode{1},\vr{x_1}) = G_2(\mode{2},\vr{x_2})$
\smallskip
\item For all $0 \leq t \leq t^*$,
$((\mode{1},\phi(t,\vr{x_1})),(\mode{2},\phi(t,\vr{x_2})) \in B$
where $t^*$ is the smallest value such that \linebreak
$G_1(\mode{1},\phi(t,\vr{x_1})) = G_2(\mode{2}\phi(t,\vr{x_2}))
= \mathit{true}$
\smallskip
\begin{spacing}{1.2}
\item For all $C \in (\modes{1} \times \mathbb{R}^{n_1})/B$, and for all
$a \in \events{1} \cup \events{2}$, \newline
$P^{\TS{}}_{1s}((\mode{1},\vr{x_1}),\{a\},C) =
P^{\TS{}}_{1s}((\mode{2},\vr{x_2}),\{a\},\match_B(C))$ and
\newline $P^{\TD{}}_{1s}((\mode{1},\vr{x_1}),\{a\},C) =
P^{\TD{}}_{1s}((\mode{2},\vr{x_2}),\{a\},\match_B(C))$.
\end{spacing}

\end{enumerate}
\end{definition}

\begin{definition}
$\TDSHA{1}$ and $\TDSHA{2}$ are \emph{labelled TDSHA bisimilar},
written $\TDSHA{1} \sim_T^l \TDSHA{2}$, whenever there exists $B$ a
TDSHA bisimulation for $\TDSHA{1}$ and $\TDSHA{2}$.  \end{definition}

We next consider the relationship between bisimulation over stochastic
HYPE models and TDSHAs.

\begin{theorem}
\label{thm:bisim}
Let 
$(P_i,\V,\IN,\IT,\Ev_d,\Ev_s,\Ac,\ec,\iv,\EC,\ID)$ for $i=1,2$, be two
stochastic HYPE models
whose TDHSAs are $\TDSHA{i}$, if $P_1 \thicksim_\mathbf{sm}^{\doteq} P_2$
then $\TDSHA{1} \sim_T^l \TDSHA{2}$.
\end{theorem}
\begin{proof} 
The stochastic HYPE models $P_i$ can be transformed as described
in Definition~\ref{def:shypetotdsha} to TDSHA $\TDSHA{i} =
(\ds(P_i),\V,\TC{i},\TD{i},\TS{i},(\cf{R_i,\sigma_i},v_i),\Ev_d
\cup \Ev_s)$. Let the equivalence relation $B$ be defined by
\[B = \bigl\{
\bigl((\cf{Q_1,\sigma_1},\vr{x}), (\cf{Q_2,\sigma_2},\vr{x})\bigr)
\mid Q_1 \thicksim_\mathbf{sm}^{\doteq} Q_2, \sigma_1 \doteq \sigma_2
\bigr\}.\]  To satisfy the first condition, let $\out_i$ be the
identity function.

Next, note that $\lambda(\cf{Q_i,\sigma_i},\vr{x}) = \sum \{
r(\cf{Q_i,\sigma_i},\sev{a},C) | \sev{a} \in \Ev_s, C \in
\CF/(\thicksim_\mathbf{sm}^{\doteq}) \}$. Since $Q_1
\thicksim_\mathbf{sm}^{\doteq} Q_2$ and $\sigma_1 \doteq \sigma_2$,
$r(\cf{Q_1,\sigma_1},\sev{a},C) = r(\cf{Q_2,\sigma_2},\sev{a},C)$.
Hence it is straightforward to show that
$\lambda(\cf{Q_1,\sigma_1},\vr{x_1}) =
\lambda(\cf{Q_2,\sigma_2},\vr{x_2})$.

For $Q_1 \thicksim_\mathbf{sm}^{\doteq} Q_2$ and $\sigma_1 \doteq
\sigma_2$, any transition that can be performed by $\cf{Q_1,\sigma_1}$
can be matched by $\cf{Q_2,\sigma_2}$ hence the same events must
be available in both configurations. $G(\cf{Q_1,\sigma_1},\vr{x})
= \bigvee \{ \act(\ev{a})(\vr{x}) \mid \cf{Q_1,\sigma_1}
\lsexrightarrow{\ev{a}} \} = \bigvee \{ \act(\ev{a})(\vr{x}) \mid
\cf{Q_2,\sigma_2} \lsexrightarrow{\ev{a}} \} =
G(\cf{Q_2,\sigma_2},\vr{x})$.

By Theorem~\ref{thm:doteqODES}, $(Q_1)_{\sigma_1} = (Q_2)_{\sigma_2}$,
for $Q_1 \thicksim_\mathbf{sm}^{\doteq} Q_2$ and $\sigma_1 \doteq
\sigma_2$ hence we can conclude that the pair
$((\cf{Q_1,\sigma_1},\phi(t,\vr{x})),(\cf{Q_2,\sigma_2},\phi(t,\vr{x})))
\in B$ where $\phi(t,\vr{x})$ is the solution of the ODEs given by
$(Q_1)_{\sigma_1}$ and for all $0 \leq t \leq t^*$ where $t^*$ is
the smallest value that $G(\cf{Q_1,\sigma_1},\vr{x})=
G(\cf{Q_2,\sigma_2},\vr{x})=\true$.

Finally, we need to show that for all $C \in (\ds(P_1) \times
\mathbb{R}^{n})/B$, and for all
$a \in \events{}$, 
\begin{itemize}
\item
$P^{\TS{}}_{1s}((\cf{Q_1,\sigma_1},\vr{x}),\{\sev{a}\},C) =
P^{\TS{}}_{1s}((\cf{Q_2,\sigma_2},\vr{x}),\{\sev{a}\},\match_B(C))$ 
\item $P^{\TD{}}_{1s}((\cf{Q_1,\sigma_1},\vr{x}),\{\ev{a}\},C) =
P^{\TD{}}_{1s}((\cf{Q_2,\sigma_2},\vr{x}),\{\ev{a}\},\match_B(C))$
\end{itemize}
Since $B$ is the identity relation over $\mathbb{R}^n$, each equivalence
class $C$ is of the form $[\cf{Q'_1,\sigma'_1}] \times \{\vr{x}\}$ and
$\match_B(C)$ is of the form $[\cf{Q'_2,\sigma'_2}] \times \{\vr{x}\}$,
where $Q'_1 \thicksim_\mathbf{sm}^{\doteq} Q'_2$.
The first item follows from the
fact that $r(\cf{Q_1,\sigma_1},\sev{a},[\cf{Q'_1,\sigma'_1}]) =
r(\cf{Q_2,\sigma_2},\sev{a},[\cf{Q'_2,\sigma'_2}])$ and reset functions
depend only on $\sev{a}$.  The second is a consequence of the fact
that $P^{\TD{}}_{1s}((\cf{Q_i,\sigma_i},\vr{x}),\{\ev{a}\},C)$ is
proportional to the number of distinct derivatives in $C$ that
$\cf{Q_i,\sigma_i}$ has after an $\ev{a}$ event. Since each
$\cf{Q_i,\sigma_i}$ must be able to match the transitions of the
other, they have the same number of distinct derivatives.  Hence
we can conclude that $B$ is a labelled TDSHA bisimulation.  \end{proof}

The converse of this theorem does not hold since $\doteq$ is defined
for a specific variable, whereas the bisimulation that is constructed
for the TDSHA can sum across multiple variables.
Consider the following simple counterexample with two one-mode
stochastic HYPE systems, with two continuous variables $X$ and $Y$, two influences
($\iota_X$ acting on $X$ and $\iota_Y$ acting on $Y$), and no events.
\smallskip

$\begin{array}{rclcrcl}
A_1 & \rmdef & \ev{init}\pc(\iota_X,a,\const) & &
B_1 & \rmdef & \ev{init}\pc(\iota_Y,b,\const) \\
A_2 & \rmdef & \ev{init}\pc(\iota_X,a+b,\const) & &
B_2 & \rmdef & \ev{init}\pc(\iota_Y,0,\const) 
\end{array}$

\smallskip
\noindent
The respective controlled systems are $P_i = A_i \syncstar B_i$
corresponding to the ODE systems $\frac{d}{dt}(X,Y) = (a,b)$ and
$\frac{d}{dt}(X,Y) = (a+b,0)$.  Looking at the TDSHA, it is easy
to see that $\TDSHA{P_1}\sim_T^\ell \TDSHA{P_2}$, as the equivalence
relation $B = \{((x,y),(x+y,0))\}$ is a TDSHA bisimulation (using
$\out_1(X,Y) = X+Y$ and $\out_2(X,Y) = X$ as output functions,
conditions 1 and 4 follow easily, while the others are trivially
true as there is no discrete jump). However, it does not hold that
$P_1 \thicksim_\mathbf{sm}^{\doteq} P_2$, as $\doteq$ requires the
ODEs in each matching mode to be the same.

\subsection{TDHSA bisimulation applied to the example}

We know consider how TDHSA bisimulation can be used for the example.
There are different ways to implement the timing of the system.
We can remove the $\Timer_i$ subcomponents from within $\Sys$ and
add the following new timer component
$\Timer \: \rmdef \: \ev{init}\pc(t,1,\const).\Timer$ with $\iv(t)
= T$,  $T$ a new variable (and without influences $t_i$ but keeping
variables $T_i$). 
Various event conditions must be modified as follows
\begin{eqnarray*}
\ec(\ev{init}) & = & (true,\rsdefa{P}{P_0} \:\wedge\: \rsdefa{T_1}{0}
\:\wedge\: \rsdefa{T_2}{0} \:\wedge\: \rsdefa{T}{0} \:\wedge\:
\rsdefa{B}{B_0} \:\wedge\: \rsdefa{M}{0}) \\
%\ec(\sev{remove}_i) & = & (\taketime_i,\;\rsdefa{P}{P\!-\!n_i} \:\wedge\:
%\rsdefa{T_i}{0}) \\
\ec(\ev{assem}_i) & = & (T \geq T'_i+\assemtime_i, \rsdefa{B}{B\!+\!m_i}) 
\end{eqnarray*}
Denote this new system by $\Assem_{T} = Sys_T
\syncstar \ev{init}.\Con$
where 
\begin{eqnarray*}
\Sys_{T}  & \rmdef &  (\Feed_1
\syncstar \Feed_2 \syncstar Feed_3) \syncstar \Inspect \syncstar \\
& & Machine_1(W_1) \syncstar \Machine_2(W_2) \syncstar \Timer. 
\end{eqnarray*}
%\commentv{These systems have random resets but our definition of TDSHA
%bisimulation does not cover this.}

We can show that $\TDSHA{}(\Assem_T) \sim^l_T \TDSHA{}(\Assem)$. In
order to do this, first observe that the LTSs of the two stochastic
HYPE models are isomorphic. This holds because in each configuration
$\cf{P,\sigma}$ the state $\sigma$ is determined by the local state
of the controlled system $P$, specifically by the local state of
the controller in $P$.  Hence, the map $\rho:\cf{\Sys \syncstar
\Con,\sigma_1} \mapsto \cf{\Sys_T \syncstar \Con,\sigma_T}$, where
$\sigma_1$ and $\sigma_T$ depend on $\Con$, is a bijection.  Given
a configuration $\cf{P,\sigma}$, we write  $\AProc_i$ in $P$ if and
only if $P = \Sys_1 \syncstar (AProc_i\syncstar\Con')$.

\begin{figure}[t]
\begin{center}
\begin{tabular}{c}
\includegraphics[width=8.5cm]{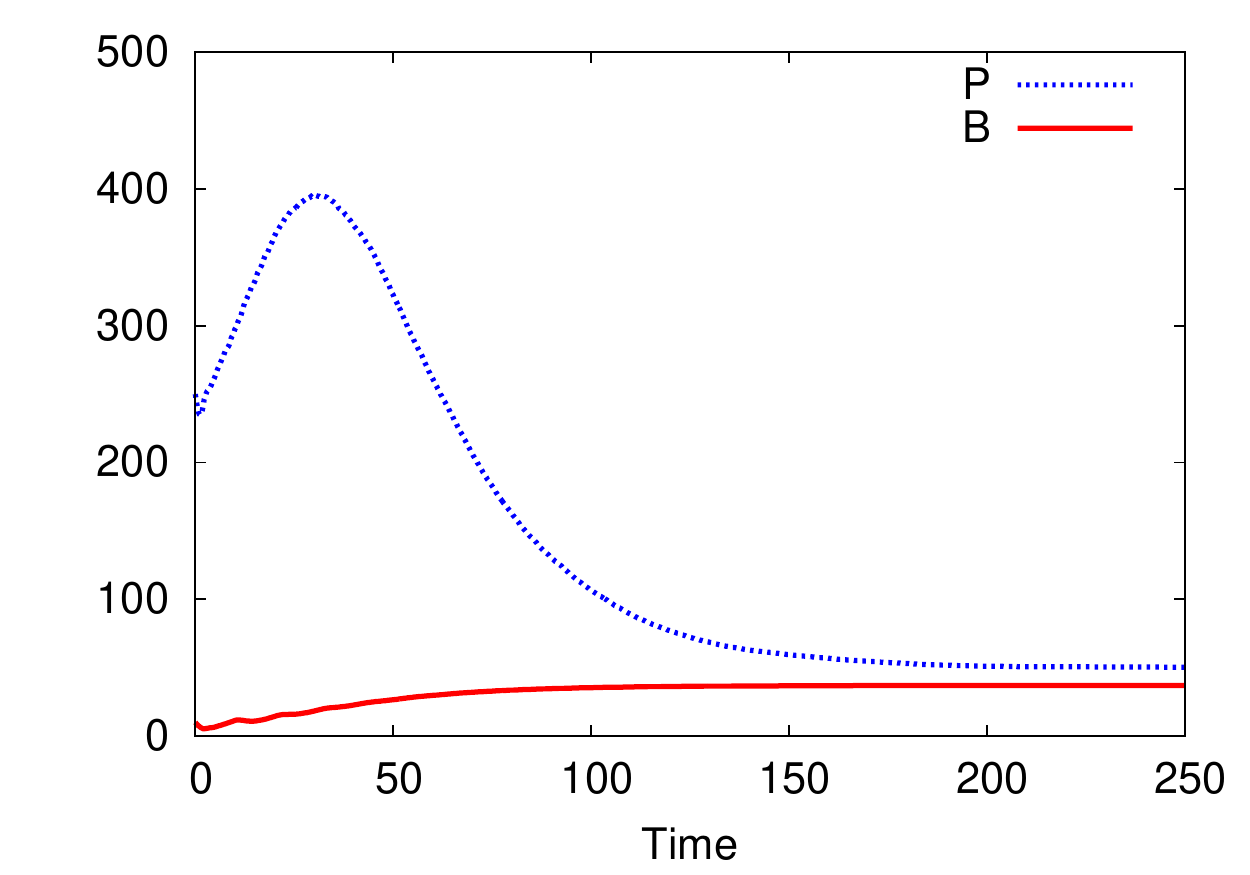}
\end{tabular}
\end{center}
\caption{Average values for $\Assem_T$ over 10000 simulations using the 
same parameters as in Figure~\ref{fig:parrun}.}
\label{fig:SimConT}
\end{figure}

When we map these stochastic HYPE models to the corresponding TDSHA,
we therefore obtain two automata with the same discrete skeleton
and the same event set, but with different variables, different
ODEs within modes, and different guards and resets for some events.
To prove that they are bisimilar, we need to exhibit a measurable
relation $B$ satisfying the conditions of Definition
\ref{def:TDSHAlabeledBisimulation}.  Let $(\cf{P,\sigma},\vr{x},s_1,s_2)\in
Q\times\bbR^6$ be a state of TDSHA $\TDSHA{}(\Assem)$, where $s_i$
is the value of timer $T_i$, and $(\rho(\cf{P,\sigma}),\vr{x},t,t_1,t_2)\in
Q\times\bbR^7$ be a state of TDSHA $\TDSHA{}(\Assem_T)$, where $t$
is the value of $T$, and $t_i$ is the value of variable $T_i$. Now,
if $\AProc_i$ in $P$ (and only in this case), then it is easy to
check that the value of $T_i$ in $\TDSHA{}(\Assem)$ has to be equal
to $T-T_i$ in $\TDSHA{}(\Assem_T)$, as both expressions measure the
time elapsed since the firing of event $\ev{remove}_i$. This suggests
the following relation
which is easily seen to be measurable.
%\footnote{This follows because the quotient sets are closed.}
\[ B = \{((\cf{P,\sigma},\vr{x},s_1,s_2),(\rho(\cf{P,\sigma}),\vr{x},t,t_1,t_2))~|~ \AProc_i\ \mathrm{in}\ P \Rightarrow t-t_i = s_i\}, \]
In order for $B$ to be a TDSHA bisimulation, we need to ignore timer
values while comparing states. This is obtained by taking the $\out$
functions to be the projections over the remaining variables:
$\out_1(\vr{x},s_1,s_2) = \vr{x}$ and $\out_T(\vr{x},t,t_1,t_2) =
\vr{x}$. Considering Definition \ref{def:TDSHAlabeledBisimulation},
condition 1 follows from the fact that the vector fields, restricted to
non-timer variables, coincide. Condition 2 is trivial, as only
instantaneous events have been modified, while the condition 3 on
guards is a consequence of the definition of $B$ in the states where
timer $i$ is active. In particular, the activation time $t^*$
coincides in both models, and so condition 4 follows. Finally,
condition 5 stems from the isomorphism of LTS and the fact that the
variables of $T_i$ are both reset after event $\ev{remove}_i$.

By contrast, the two HYPE models $\Assem$ and $\Assem_T$
are trivially not system bisimilar. This is easily seen by inspecting
Definition \ref{def:systemBisim}, which requires variable sets and
event condition to be the same in both models.

%\noindent\emph{COMMENT: this can be shown to be bisimilar at the
%TDSHA level if one ignores $T$, $T_1$ and $T_2$, but is more difficult
%ito formalise at the stochastic HYPE level because of the modified
%event conditions.}
\end{appendix}
\end{document}